\documentclass[a4paper,11pt]{article}

\usepackage{graphicx}
\usepackage{url}
\usepackage{csquotes}
\usepackage[labelsep=endash, figurename=Figure, tablename=Table, textfont=it]{caption}
\usepackage{subcaption}
\usepackage{float}  
\usepackage{hyperref}
\usepackage{multicol}
\usepackage{multirow}
\usepackage{tikz}

\usepackage{fancyhdr}
\usepackage[top=2.5cm, bottom=2.5cm, left=2.5cm, right=2.5cm]{geometry} 
\usepackage[usenames,dvipsnames]{pstricks}
\usepackage[round,comma,authoryear]{natbib}
\usepackage{titling}
\usepackage{amsmath,amssymb, mathtools,authblk, amsthm}
\usepackage[dvipsnames]{xcolor}
\usepackage[sc]{mathpazo} 
\usepackage{comment}
\usepackage[linesnumbered,ruled,vlined]{algorithm2e}
\usepackage{booktabs}

\definecolor{darkblue}{rgb}{0,0,0.5}
\hypersetup{colorlinks=true,linkcolor=black,citecolor=darkblue,urlcolor=darkblue} 

\DeclareMathOperator*{\argmin}{argmin}

\newtheorem*{proposition*}{Proposition}
\newtheorem{theorem}{Theorem}
\newtheorem{lemma}{Lemma}

\newtheorem{assumption}{Assumption}

\newcommand{\R}{\mathbb{R}}
\def\tr{\color{red}}

\begin{document}

\title{Sensitivity analysis of the\\ perturbed utility stochastic traffic equilibrium\thanks{We are grateful to an anonymous referee for pointing out the connection to the machine learning literature. We thank  Jesper R.-V. Sørensen for catching a mathematical bug, the correction of which led us to o-minimal geometry.}}
\author[a]{Mogens Fosgerau\footnote{Corresponding author: \url{mogens.fosgerau@econ.ku.dk}}}
\author[a]{Nikolaj Nielsen}
\author[b]{Mads Paulsen}
\author[b]{Thomas Kjær Rasmussen}
\author[c]{Rui Yao}
\affil[a]{University of Copenhagen, Denmark}
\affil[b]{Technical University of Denmark, Denmark}
\affil[c]{Technion, Israel}
\date{\today}
\maketitle

\begin{abstract}
\textbf This paper develops a sensitivity analysis framework for the perturbed utility route choice (PURC) model and the accompanying stochastic traffic equilibrium model. We derive analytical sensitivity expressions for the Jacobian of the individual optimal PURC flow and equilibrium link flows with respect to link cost parameters under general assumptions. This allows us to determine the marginal change in link flows following a marginal change in link costs across the network. We show how to implement these results while exploiting the sparsity generated by the PURC model. 
Numerical examples illustrate the use of our method for estimating equilibrium link flows after link cost shifts, identifying critical design parameters, and quantifying uncertainty in performance predictions. Finally, we demonstrate the method in a large-scale example. The findings have implications for network design, pricing strategies, and policy analysis in transportation planning and economics, providing a bridge between theoretical models and real-world applications. 

\end{abstract}

\bigskip


\section{Introduction}
This paper presents analytical results and computational algorithms for sensitivity analysis of the equilibrium assignment with the perturbed utility route choice (PURC) model. This allows estimation of the change in equilibrium flows that follows a change in network parameters, without having to resolve for equilibrium. 

More specifically, we provide analytical expressions for the following. 
\begin{itemize}
    \item The Jacobian of the predicted flow, in the case where link costs are constant. 
    \item The Jacobian of equilibrium \textit{link costs} with respect to network parameters, in the case where link costs are flow-dependent.
    \item The Jacobian of equilibrium \textit{link flows} with respect to network parameters, in the case where link costs are flow-dependent.
\end{itemize}

To obtain these results, we have overcome some specific challenges presented by the PURC model. It is an important and attractive feature of the PURC model that it uses the whole network as choice set, but may predict that network links are inactive, having zero flow. This means we must deal with a flow conservation constraint at each network node as well as potential non-differentiability at points where some links change between being active and inactive. 

With the analytical expressions for the Jacobians, we can approximate equilibrium outcomes in new scenarios without having to resolve for equilibrium. This is of value if we are able to compute the Jacobians faster than we can recompute the equilibrium. We provide algorithms for this purpose that exploit the sparsity that occurs when many network links in a large network are inactive.

These results have many potential applications. One is the economic appraisal of changes in network parameters. The perturbed utility framework implies a natural welfare measure, namely the maximized perturbed utility or, as it is often called, the value function. We demonstrate this is continuously differentiable as a function of link costs, with a gradient that is the link flow. This observation allows our sensitivity results to be applied directly for approximate cost-benefit analysis. 

Perhaps the most important application of sensitivity analysis is to bilevel optimization problems, e.g., network design and pricing problems \citep{yang1997traffic,josefsson2007sensitivity}, (hyper)parameter optimization and meta-learning \citep{franceschi2018bilevel,lorraine2020optimizing}, as well as deep equilibrium and implicit models \citep{bai2019deep}. Such problems are generally computationally intensive and require numerous model solutions. Our results may be used to speed up computation, thereby increasing the size of bilevel problems that can be addressed in practice. 

We set up the PURC model and the perturbed utility stochastic traffic equilibrium model in Section \ref{sec:PURC}. Our theoretical results are presented in Section \ref{sec:theoretical_results}. Section \ref{sec:numerical} illustrates the use of our theoretical results in numerical examples using toy networks. We present a prediction of equilibrium link flow patterns following changes in network parameters, without needing to resolve for equilibrium. We also show how sensitivity results can assist in the uncertainty analysis of prediction outputs when there is uncertainty about network parameters. A particular finding, illustrated with examples, is that some pairs of paths can be complements in the PURC model. This contradicts a result in a previous paper in this journal. Although it is not a focus of this paper, we take the opportunity to point out and correct the error. Finally, Section \ref{sec:large_scale} applies our main results using a real-world network. Section \ref{sec:conclusion} concludes. The following section provides a literature review.

\subsection{Literature review}

%\textbf{The literature review could then be a section on its own, remind the readers about PURC and PU equilibrium model (already published), and then focus on the sensitivity analysis and its importance. }

The perturbed utility route choice (PURC) model \citep{fosgerau_perturbed_2022, fosgerau_bikeability_2023} predicts the route choice behavior of an individual traveler as a link flow vector across a network that solves a convex minimization problem, defined in terms of link costs and a convex perturbation function of the link flow vector. The flow vector is constrained to satisfy flow conservation from origin to destination. Thus, the PURC model belongs to the general class of perturbed utility models \citep{fosgerau2012theory,allen2019identification}, adapted to the route choice context. The perturbation function plays an important role for two reasons. First, it directly incorporates the network structure. This determines the substitution patterns of the PURC model, which therefore take the network structure into account, without the need to extend the model to explicitly account for "correlation". Second, the perturbation function is constructed to allow corner solutions, such that the model can predict travel only on a small active subnetwork, specific to each origin-destination pair. 

The corresponding perturbed utility traffic equilibrium model is developed in \cite{yao_perturbed_2024}, who develop a fast algorithm for computing the equilibrium assignment for large problems. The basis for this is a mathematical trick, due to \citet{beckmann1956studies}, which formulates the equilibrium as the solution to a convex optimization problem involving a certain potential function. We employ the same potential function in this paper. 

We now discuss sensitivity analysis for traffic equilibrium. This quantifies the rates of change in \textit{equilibrium} traffic flows as a function of network parameters.

A key application of traffic equilibrium sensitivity is simply to estimate the change in equilibrium flow patterns that follows a change in network parameters without having to completely re-solve the traffic equilibrium problem \citep{tobin1988sensitivity, yang1997traffic, patriksson2003sensitivity}. This, in turn, is important for the welfare analysis of transportation network interventions \citep{Small1992}.

More broadly, sensitivity analysis is central in bilevel optimization problems~\citep{luo1996mathematical, josefsson2007sensitivity}, where a traffic equilibrium appears as the lower-level constraint.  
Two classic examples are capacity expansion and road pricing. In the capacity expansion problem~\citep{yang1998models, gao2004continuous}, sensitivity information points to the direction in which marginal investments have the largest impacts on equilibrium flows, guiding the design of efficient network investments. 
In the pricing problem~\citep{yan1996optimal, wang2021optimal}, sensitivity with respect to link costs enables the calculation of gradients for objectives such as social welfare or revenue, which is necessary for designing general pricing schemes beyond classic marginal pricing. 
Thus, an important motivation for sensitivity analysis is to make large-scale bilevel optimization tractable. 

Existing methods for computing sensitivities often require solving auxiliary equilibrium problems~\citep{patriksson_sensitivity_2004,josefsson2007sensitivity, lu2008sensitivity} or inverting large dense matrices~\citep{ying2005sensitivity, clark2006applications,yang2009sensitivity}, both of which are computationally demanding and difficult to scale to real-size networks. 
This scalability issue motivates the analytical results and sparse algorithms that we develop in this paper.

The technical foundation for sensitivity analysis is the differentiability of the equilibrium solution with respect to network parameters. Classical results \citep{robinson1985implicit, robinson2006strong} show that if the equilibrium conditions are continuously differentiable with an invertible Jacobian, sensitivities can be obtained via the implicit function theorem~\citep{clarke1990optimization}.
Along these lines, \cite{tobin1988sensitivity}, \cite{yan1996optimal}, and \cite{cho2000reduction} propose deriving equilibrium sensitivity on reduced networks using the implicit function theorem, under additional differentiability assumptions.
However, differentiability is not always guaranteed~\citep{patriksson_sensitivity_2004}. In deterministic equilibrium models, for instance, flows may concentrate on a limited set of minimum-cost paths, so that small parameter changes can shift the set of active links in a non-smooth way~\citep{patriksson2003sensitivity}.
Consequently, directional differentiation is required for equilibrium sensitivity analysis to network parameters~\citep{patriksson_sensitivity_2004, josefsson2007sensitivity}, and elastic demands~\citep{lu2008sensitivity}. 
Stochastic user equilibrium models based on additive random utility~\citep[e.g.,][]{dial1971probabilistic,maher1997probit,bekhor2001stochastic, baillon2008markovian, kitthamkesorn2014unconstrained, oyama2022markovian} avoid the issue of differentiability by assigning strictly positive flows to all links or paths, but at the cost of losing the ability to predict zero flows on long detour paths. The PURC model combines both features: it predicts equilibrium link flow patterns where travelers are predicted to use not only the minimum cost paths while still predicting zero flows on long detour paths \citep{fosgerau_perturbed_2022,fosgerau_bikeability_2023,yao_perturbed_2024}. The possibility of zero flows raises the issue of differentiability. We establish that the equilibrium link flows in the PURC framework are \textit{differentiable almost everywhere}, and derive analytical sensitivity expressions based on the projected first-order condition, which can be evaluated on the active subnetwork and thus supports efficient sparse computation. 
Our analytical sensitivity results are consistent with classic works such as~\cite{tobin1988sensitivity}, but our analysis and proof exploit specific properties of the PURC. We detail our analysis in the following sections.

\section{The perturbed utility route choice and traffic equilibrium models\label{sec:PURC}}

\subsection{PURC model setup}\label{sec:PURC model}

The PURC model~\citep{fosgerau_perturbed_2022, fosgerau_bikeability_2023} is defined for a network $(\mathcal{N},\mathcal{L})$, where $\mathcal{N}$ is the set of nodes with typical element $v$ and $\mathcal{L}$ is the set of directed links with typical element $ij$ for a link from node $i$ to node $j$. We assume at least one path exists between any two network nodes, i.e.,

\begin{assumption}
\label{A:1} The network $(\mathcal{N},\mathcal{L})$ is connected.
\end{assumption}

The node-link incident matrix $A\in \mathbb{R}^{|\mathcal{N}| \times |%
\mathcal{L}|}$ has entries 
\begin{align*}
a_{v, ij}=%
\begin{cases}
-1,& v=i, \\ 
1,& v=j, \\ 
0,& \text{otherwise}.%
\end{cases}%
\end{align*}
The demand for a traveler is encoded as a unit demand vector $b\in \mathbb{R} ^{|\mathcal{N}|}$ with components
\begin{align*}
b_v = 
\begin{cases}
-1, & \text{if $v$ is the traveler's origin}, \\ 
1, & \text{if $v$ is the traveler's destination}, \\ 
0, & \text{otherwise}.%
\end{cases}%
\end{align*}
A network link flow vector $x\in \mathbb{R}_+^{|\mathcal{L}|}$ satisfies flow conservation if $Ax=b$. The network link lengths are summarized in the vector $l = \{l_{ij}\}_{ij \in \mathcal{L}}\in \mathbb{R}_{++}^{|\mathcal{L}|}$.

A traveler in the PURC model is associated with a demand vector $b$, a vector of non-negative link costs $c\in \mathbb{R}_{+}^{|\mathcal{L}|}$, and link-specific convex perturbation functions $\{F_{ij}\}_{ij \in \mathcal L}$, satisfying the following assumption.

\begin{assumption}
\label{A:2} A link perturbation function, $F_{ij}:\mathbb{R}_+ \rightarrow \mathbb{R%
}_+,  ij \in \mathcal{L}$, is twice continuously differentiable, strictly convex, and strictly increasing with $F_{ij}(0)=F^{\prime }_{ij}(0)=0$, $F^{\prime \prime}_{ij}(x_{ij})>0$ for $x_{ij}\geq 0$, and range equal to $\mathbb{R}_+$. Define $(F^{\prime }_{ij})^{-1} (y)=0$ for $y<0$ such
that the inverse function of the derivative of the perturbation function has domain equal to $\mathbb{R}$.
\end{assumption}

For a flow vector $x\in \mathbb{R}_{+}^{|\mathcal{L}|}$, we define 
\begin{equation}\label{eq:sum_link_perturbations}
F(x)=\sum_{ij\in \mathcal{L}}F_{ij}(x_{ij})
\end{equation}%
for the sum across links of perturbations $F_{ij}(x_{ij})$. The link perturbations may incorporate link lengths $l_{ij}$ as in the examples below. Eq.~\eqref{eq:sum_link_perturbations} generalizes \citet{fosgerau_perturbed_2022} by allowing perturbations $F_{ij}(x_{ij})$ to be link-specific.

Specific instances of the perturbation function for use in applications include, e.g., a specification based on entropy 
\begin{equation}
F_{ij}\left( x_{ij}\right) = l_{ij} \left(\left( 1+x_{ij}\right) \ln \left( 1+x_{ij}\right)
-x_{ij}\right)  \label{eq:Fentropic}
\end{equation}%
and a quadratic specification 
\begin{equation}\label{eq:Fquad}
F_{ij}\left( x_{ij}\right) =l_{ij} x_{ij}^{2},
\end{equation}
where $l_{ij}$ is the length of link $ij$.

The PURC model predicts an individual's routing decisions as a network flow vector $x^*$ that solves an individual-level cost minimization problem of the perturbed utility form 
\begin{subequations}
\label{eq:PURC}
\begin{align}
\min_{x \geq 0}&\, c^{\top }x+F(x)
\label{eq:PURC_objective} \\
\text{s.t.}\quad & Ax=b.  \label{eq:PURC_onservation}
\end{align}
\end{subequations}
The constraint $Ax=b$ implies conservation of network flow. 

Thus, the PURC model is a utility-maximizing model, which ensures that the model predictions conform to standard rationality requirements. The PURC model connects to data with the assumption that the optimal flow vector $x^*$ is the expectation of the vector $y$, which is observed path choice encoded as a vector of zeros and ones, with ones indicating which links are included in the observed path.

As a utility-maximizing model, the PURC model admits a natural welfare measure. Let $W$ denote the negative value function,
\begin{eqnarray}\label{eq:W}
    W(v)=- \min_{Ax=b, x\geq 0}\left\{ x^\top (-v)+F(x)\right\},
\end{eqnarray}
where $v=-c$. This is the optimal perturbed utility that a traveler obtains. Theorem \ref{thm:jacobian} below presents the gradient of the value function, and the use of this result for welfare analysis of marginal changes is discussed in Section \ref{sec:welfare}.

We now derive the optimal link flow vector  $x^*$ for a given cost vector $c$ and demand $b$, which is the key ingredient to our sensitivity analysis of the PURC. We start with the Lagrangian for the individual traveler's perturbed cost minimization problem: 
\begin{equation}  \label{eq_Lagrange}
\Lambda( x,\eta) =x ^\top c + F(x) +\eta ^\top \left( Ax-b\right) ,\; x\in 
\mathbb{R}_+^{|\mathcal{L}|}, \eta\in \R^{|\mathcal{V}|}.
\end{equation}

Let $B^* = \mathrm{diag} \left( 1_{\{x^*>0\}} \right)$ be the matrix with ones on the diagonal corresponding to positive optimal link flows $x^*_{ij}$. The first-order condition for the cost minimization problem can then be written as
\begin{equation}
{B^*}\left( c+\nabla F\left( x^*\right) +A^{\top }{\eta^*}\right) =0,
\label{eq:1}
\end{equation}%
where $\nabla F(x^*) = \{F_{ij}^{\prime
}(x^*_{ij})\}_{ij\in \mathcal{L}}$ is the gradient of $F$.

Inspecting Eq. \eqref{eq:1}, we see that it comprises dual variables $\eta^*$ corresponding to the flow-conservation constraints at each node in the network.  We shall proceed to eliminate these. Let ${P^*}$ be the matrix 
\begin{equation}\label{eq:Pstar}
{P^*}={B^*}-(AB^*)^+A B^*,
\end{equation}%
where $(A{B^*})^{+}$ is the Moore-Penrose inverse of $A B^*$. Lemma \ref{lem:P_hat} in the Appendix shows that $ P^*$ is the orthogonal projection onto the linear subspace
\begin{eqnarray}\label{eq:subspace}
\{ x\in \mathbb R^{|\mathcal{L}|}:A x =0, B^* x =x\},
\end{eqnarray}
i.e. the intersection of the null space of $A$ and the subspace where the components corresponding to inactive links are zero. Pre-multiplying Eq. \eqref{eq:1} by ${P^*}$ yields (see again Lemma \ref{lem:P_hat}) the projected first-order condition
\begin{equation}\label{eq:projectedFOC}
{P^*}(c+\nabla F(x^*))=0.
\end{equation}

\subsection{Perturbed utility stochastic traffic equilibrium model}
The perturbed utility stochastic traffic equilibrium model~\citep{yao_perturbed_2024} extends PURC to the case where the collective routing decisions of individuals influence link costs. Specifically, the stochastic traffic equilibrium model formulates the equilibrium problem as an optimization problem:
\begin{subequations}
\label{eq:PURC_SUE}
\begin{align}
\min_{x \geq 0}&\, \sum_{{ij} \in \mathcal{L}}\int_0^{\sum_{w\in\mathcal{W}} q^w x_{ij}^w} \zeta _{ij}(u)du+ \sum_{w\in\mathcal{W}} q^w F^w(x^w) \label{eq:PURC_SUE_obj}\\
\text{s.t.}\quad & (Ax^w - b^w) = 0, \; \forall w \in \mathcal{W},
\end{align}
\end{subequations}
where $\mathcal{W} = \{w\}$ is the set of traveler \textit{types}, $x^w = (x_{ij}^w)_{{ij}\in\mathcal{L}}$ is the set of link flows for type $w$,  $x= \sum_{w\in \mathcal{W}} q ^w x^w$ is the aggregate link flow vector, $\zeta_{ij}(\cdot)$ is the link cost function, $q^w \in \mathbb{R}_{++}$ is the total demand for each type $w$, $F^w$ is the perturbation function for type $w$, and $b^w$ is a unit demand vector that encodes the OD for type $w$. With these definitions, the term $q^w x_{ij}^w$ is the flow on link $ij$ of travelers of type $w$. 

The equilibrium problem has a unique solution if link cost functions are monotone~\citep{yao_perturbed_2024}. Hence, we make the following assumption on link cost functions, which will be convenient for our sensitivity analysis presented in the next section.

\begin{assumption}\label{A:3}
    The link cost function $\zeta_{ij}:\mathbb{R} \rightarrow \mathbb{R}_{++}, \forall {ij} \in \mathcal{L}$ is a continuously differentiable,  strictly increasing function of aggregate link flows $x_{ij} \in \mathbb{R}$.
\end{assumption}

Assumption~\ref{A:3} implies that the link cost functions have inverses and hence that link flows corresponding to any cost $c_{ij}$ can be expressed in terms of the inverse cost function:
\begin{align}
    \zeta_{ij}^{-1}(c_{ij}) = x_{ij}, \; \forall {ij} \in \mathcal{L}.
\end{align}

Therefore, we can write the equilibrium condition as 
\begin{equation}\label{eq:prop_dual_fixed_point}
        \zeta_{ij}^{-1}(c^*_{ij}) = x^*_{ij}(c^*), \; \forall ij \in \mathcal{L},
    \end{equation}
where $x^*_{ij}(c) = \sum_{w\in\mathcal{W}} q^w x_{ij}^{w*}(c)$ is the aggregate PURC optimal flow on link $ij$ given cost $c$. We will use this relationship to connect the sensitivity analysis of traffic equilibrium to the sensitivity analysis of the PURC optimal flows.

Our setup thus allows heterogeneous traveler types. To unburden the notation, we will not carry the potential heterogeneity of the perturbation function through the analysis. Hence, from this point, we will assume $F^w=F,w\in \mathcal{W}$. 

We note that it is straightforward to allow heterogeneous link costs of an additively separable form, $\tilde \zeta_{ij}^w=c^w_{ij} + \zeta_{ij}$.

Differentiating the objective~\eqref{eq:PURC_SUE_obj} with respect to type-specific link flows $x_{ij}^w$, we obtain the type-specific perturbed link cost as 
\begin{equation*}
    q^w(\zeta_{ij}(x_{ij}) +\nabla F_{ij}(x_{ij}^w)),
\end{equation*} 
where the first term in the parenthesis reflects congestion due to collective routing decisions $x_{ij} = \sum_{w\in\mathcal{W}} q^w x_{ij}^w$, and the second term is the type-specific (marginal) perturbation cost $\nabla F_{ij}(x_{ij}^w)$, representing, for example, travelers' taste for variety.

\section{Sensitivity analysis}\label{sec:theoretical_results}

\subsection{Flow-independent case: perturbed utility route choice}\label{sec:main_result}

The PURC model has the realistic feature that many links may have zero flow, i.e.\ they may be inactive. This implies that flow can change in a non-smooth way as links may change between being active and inactive as the cost vector changes. To help handle that issue, we define the \textit{activation boundary},
\begin{eqnarray}\label{eq:C0}
    \mathcal{C}_0=\left\{ c:1_{x^*(c)>0} \textrm{\; is not continuous at $c$} \right\},
\end{eqnarray}
as the set of cost vectors $c$ where a marginal change in costs may cause the set of active links to change.

Using the projection matrix $ P^*$ defined in Eq. \eqref{eq:Pstar}, we now establish our first main result, showing that the optimal flow is equal to the gradient of the value function. Outside the activation boundary, we show the optimal flow is continuously differentiable, and we give an analytical expression for the Jacobian of the optimal flow vector differentiated with respect to the link cost vector. This Jacobian comprises all information about how demands respond to marginal cost changes.

Finally, we show that the activation boundary is small, of Lebesgue measure zero. For this result, we require a regularity assumption, namely that the perturbation function  $F$ is \emph{definable} \citep{van1998tame}. Functions defined in terms of polynomials, exponentials and logarithms are definable, and hence the entropic \eqref{eq:Fentropic} and quadratic \eqref{eq:Fquad} perturbation functions are examples of definable functions. We give a brief introduction to definability in Appendix \ref{app:definable}.

\begin{theorem}\label{thm:jacobian} Under Assumptions \ref{A:1} and \ref{A:2}, 
\begin{enumerate}
    \item  The value function is continuously differentiable and satisfies
\begin{eqnarray}\label{eq:Roys}
    \nabla W(v)= x^*(-v).
\end{eqnarray}
In particular, $x^*$ is continuous.
    \item $x^*$ is continuously differentiable on the complement of $\mathcal{C}_0$ with symmetric negative semidefinite Jacobian
    \begin{eqnarray}\label{eq:eqngradient}
        \nabla x^*(c)=-\left(P^*\nabla^2F(x^*)P^*\right)^+.
    \end{eqnarray}
    \item If $F$ is definable, then $\mathcal{C}_0$ has Lebesgue measure zero.
\end{enumerate}
\end{theorem}
\begin{proof}
 \textbf{Item 1.} Let $\mathcal{D}$ denote the feasible set $\mathcal{D}=\left\{x: Ax=b,\; x\geq 0\right\}$ and let $G(x)=F(x)$ for $x\in \mathcal{D}$ and $G(x)=\infty$ otherwise. Then with $v=-c$,
    \begin{eqnarray}
        W(v)=\sup_{x\in \mathbb R^{|\mathcal{L}|}}\left\{x^\top v-G(x)\right\}.
    \end{eqnarray}
A sufficient condition for a function $G$ to be \textit{essentially strictly convex} \citep{rockafellar_convex_1970} is that it is strictly convex on every convex subset of its domain $\mathcal{D}$. By Assumption \ref{A:2}, $F$ is strictly convex on $\R^{|\mathcal{L}|}_+$, such that $G$ is strictly convex on $\mathcal{D}$ since $\mathcal{D}\subset \R^{|\mathcal{L}|}_+$ and $F=G$ on $\mathcal{D}$. 

A convex function $f:\R^d\rightarrow \R$ is \textit{essentially smooth} \citep{rockafellar_convex_1970} if \textit{i)} the interior of its domain is nonempty, \textit{ii)} it is continuously differentiable on the interior of its domain and \textit{iii)} $\underset{k\rightarrow \infty}{\lim} ||\nabla f(z_k)||=\infty$ for any sequence $(z_k)$ in the interior of the domain converging to a boundary point of the domain. Theorem 26.3 in \citet{rockafellar_convex_1970} then establishes that $W$ is essentially smooth, since $W$ is the convex conjugate of $G$. 

Since $W$ is essentially smooth, it is continuously differentiable on the interior of its domain. The interior of the domain includes the negative orthant since for any $v<0$,
\begin{eqnarray}
    W(v)=\sup_{x\in \mathcal{D}}x^\top v-F(x)\leq \sup_{x\geq 0} x^\top v -F(x)\leq \sup_{x\geq 0}-F(x) =0,
\end{eqnarray}
where we have used that $F(x)\geq 0$ and that $x^\top v<0$ for $x\geq 0$ and $v<0$. Applying \citet{milgrom_envelope_2002} Theorem 1 to the directional derivatives of $W$ then shows that $\nabla W(v)=x^*(-v)=x^*(c)$ for $v=-c<0$.

% \begin{eqnarray*}
%  -\infty< \inf_{x\in [0,1]^{|\mathcal{L}|}}x^\top v-F(x)\leq  W(v)\leq \sup_{x\in [0,1]^{|\mathcal{L}}}x^\top v-F(x)<\infty,
% \end{eqnarray*}
% where the infimum and supremum are finite by continuity over a compact set. 

\textbf{Item 2.} By Alexandrov's Theorem \citep{Alexandroff1939}, $W$ is twice differentiable almost everywhere. Since $\nabla W(v)=x^*(-v)$ for all $v\in \R^{|\mathcal{L}|}_-$, we have that $x^*$ is differentiable almost everywhere with symmetric negative semidefinite Jacobian $\nabla x^*(-v)=-\nabla^2 W(v)$.

Now, let $c \not \in \mathcal C_0$ where $\mathcal{C}_0$ was defined in \eqref{eq:C0}. It follows from the projected first order condition \eqref{eq:projectedFOC} that $x^*(P^*c)=x^*(c)$. 
Let $ij$ denote a link such that $x^*_{ij}(c)=0$. Since $(P^*c)_{ij}=0$, it follows that $x^*(c)$ is differentiable with respect to $c_{ij }$ with $\frac{\partial x^*(c)}{\partial c_{ij}}=0$. By assumption $c\not \in \mathcal{C}_0$, such that $1_{x^*(c)>0}$ is constant in a neighborhood of $c$. As $ij$ is a link with $x^*_{ij}(c)=0$ it follows that $x^*_{ij}(c')=0$ for any $c'$ sufficiently close to $c$. It then follows that $\frac{\partial x^*_{ij}(c)}{\partial c}$ exists and equals zero. It remains to be shown that $\frac{\partial x^*_{ij}(c)}{\partial c_{i'j'}}$ exists when $ij$ and $i'j'$ are both active links. 

Given the above, we focus on the nonzero flows as a function of the corresponding costs. To simplify notation, we therefore assume all links are active. The optimal flow $x^*(c)$ is then the solution to $
    \min_{x,\eta} x^\top c+F(x)+\eta^\top (Ax-b)$ which by Lemma \ref{lem:fullrank} is equivalent to $    \min_{x,\eta} x^\top c+F(x)+\eta^\top(Cx-d)$, 
where $CC^\top$ is invertible. $x^*(c)$ is therefore the unique solution to the minimization problem
\begin{eqnarray}
    x^*(c)=\arg \min_{Cx=d} x^\top c+F(x)
\end{eqnarray}
with Lagrangian
\begin{eqnarray}
    L(x,\eta)=x^\top c+F(x)+\eta^\top(Cx-d)
\end{eqnarray}
and KKT conditions
\begin{eqnarray}
    c+\nabla F(x)+C^\top \eta &=& 0 \\
    Cx-d&=&0.
\end{eqnarray}
Let $z=(x^\top,\eta^\top)^\top\in \R^{|\mathcal{L}|+|\mathcal{V}|}$ and define
\begin{eqnarray}
    \Phi(z,c)=\begin{pmatrix}
        c+ \nabla F(x)+C^\top \eta\\
    Cx-d\end{pmatrix}.
\end{eqnarray}
Then
\begin{eqnarray}
    \nabla_z \Phi(z,c)=\begin{pmatrix}
        \nabla^2 F(x) & C^\top \\
        C & 0
    \end{pmatrix}.
\end{eqnarray}

% Now, consider a $v=\begin{pmatrix}
%     v_1\\ v_2
% \end{pmatrix}$ with $v\in \R^{|\mathcal{L}|+r}$ and $v\neq 0$. Then if $\nabla_Z \Phi(z,c)v=0$ we have
% \begin{eqnarray}
%     0=\nabla_z \Phi(z,c)v = \begin{pmatrix}
%         \nabla^2 F(x)v_1 +C^\top v_2\\
%         Cv_1 
%     \end{pmatrix},
% \end{eqnarray}
% such that in particular $Cv_1=0$. But then $v^\top \Phi(z,c)v=v_1^\top \nabla^2 F(x)v_1+v_2^\top Cv_1 = v_1^\top \nabla^2 F(x)v_1 =0$ implying that $v_1=0$. But then $\nabla_z \Phi(z,c)v=0$ simplifies to $C^\top v_2=0$, which implies $v_2=0$ since $C$ has full rank. We conclude that $\nabla_z \Phi(z,c)v=0$ if and only if $v=0$, such that $\nabla_z \Phi(z,c)$ is invertible. 

By the properties of the Schur complement,
\begin{eqnarray}
    \det\left(\nabla_z \Phi(z,c)\right)= \det(\nabla^2 F(x) )\det\left(- C [\nabla^2 F(x)]^{-1}C^\top\right)\neq 0,
\end{eqnarray}
since $C [\nabla^2 F(x)]^{-1} C^\top$ has full rank. By construction, for any $v\in \mathbb R^{|\mathcal{V}|}$ with $v\neq 0$ we have $C^\top v\neq 0$, such that in particular $v^\top C\nabla^2F (x)^{-1} C^\top v = (C^\top v)^\top \nabla^2 F(x)^{-1} (C^\top v)$, which is nonzero since $\nabla^2 F(x)^{-1}$ is positive definite and $C^\top v\neq 0$. It now follows that $z^*(c)$ is continuously differentiable in a neighborhood of $c$ by the implicit function theorem \citep[][Theorem 3.3.1]{krantz_implicit_2003}.

We now derive the expression for $\nabla x^*(c)$. Differentiating both sides of the projected first-order condition \eqref{eq:projectedFOC} yields $P^* \nabla^2 F(x^*) \nabla_c x^*(c)=- P^*$,
where we have determined already that $P^*\nabla^2 F(x^*)=P^*\nabla^2 F(x)^*P^*$, such that this becomes
\begin{eqnarray}\label{eq:diff_projected_foc}
   P^* \nabla^2 F(x^*)P^* \nabla_c x^*(c)=- P^*.
\end{eqnarray}
For ease of notation, let $X=\nabla x^*(c)$ and $Y=-P^*\nabla^2 F(x^*)P^*$. Then we have that $X,Y$ are symmetric and $P^*Y=Y$. Further, by differentiating both sides of constraint $A B^* x^* = b$ with respect to $c$, we have that $A B^* X = 0$, implying that $P^* X = X$. Eq.~\eqref{eq:diff_projected_foc} becomes $YX=P^*$, which shows that $YX$ is symmetric. Transposing both sides yields $XY=P^*$ such that $XY$ is symmetric. By pre- and postmultiplication respectively, we find  $XYX=X$ and $YXY=Y$ such that $X$ is the Moore-Penrose inverse of $Y$, i.e.
\begin{eqnarray}\label{eq:derivative}
    \nabla x^*(c)=-\left( P^*\nabla^2 F(x^*)P^*\right)^+,
\end{eqnarray}
as claimed.

%\textbf{Item 3.} Since $x^*$ is differentiable almost everywhere and $x^*$ is differentiable on the complement of $\mathcal{C}_0$, it follows that $\mathcal{C}_0$ has Lebesgue measure zero.  

\textbf{Item 3.}
By Item 1, $x^*$ is the unique solution to the PURC problem \eqref{eq:PURC}, and each coordinate function $x^*_l$ is continuous. Then each set 
\[
S_l \coloneqq \{c\in\mathbb R^{|\mathcal L|}:\ x^*_l(c)>0\},l\in \mathcal L
\]is open. Since $S_l$ is open, the indicator function $ 1_{x^*_l(\cdot)}$ is discontinuous at $c$ iff $c\in \mathrm{bd}(S_l)$. 
The set of points of discontinuity of the multivariate indicator function defining the activation boundary $1_{x^*(\cdot)>0}$ is exactly the set of points where one coordinate function $ 1_{x^*_l(\cdot)}$ is discontinuous, hence 
\[
\mathcal{C}_0=\bigcup_{l\in \mathcal{L}}\mathrm{bd}(S_l).
\]
Thus, it suffices to show that each $\mathrm{bd}(S_l), l\in \mathcal L$ is a  Lebesgue null set.

By Lemma \ref{lem:definable}, the optimal flow function $x^*$ is definable. Hence, each set $S_l$ is definable. Hence, the boundary  $\mathrm{bd}(S_l)=\mathrm{cl}(S_l)\setminus S_l$ is definable \citep[][Prop. 1.12]{coste1999introduction}.

By the cell decomposition theorem \citep[][Thm. 6.6]{coste1999introduction}, any definable set $M \in \mathbb{R}^{|\mathcal{L}|}$ is the finite union of cells, where each cell is a $C^1$ submanifold on $\mathbb{R}^{|\mathcal{L}|}$~\citep[][Ch. 6.2]{coste1999introduction}. The o-minimal dimension $\dim_{\mathrm{om}}$~\citep[][Ch. 4]{van1998tame} is defined as
\[
\dim_{\mathrm{om}}(M) = \max_i \dim(\mathrm{cell}_i),
\]
where $\dim(\mathrm{cell}_i)$ is the topological dimension of $\mathrm{cell}_i$. Our motivation for imposing definability is that the boundary of a definable set has lower dimension  \citep[][  Corollary 1.10]{van1998tame}: 
\[
\dim_{\mathrm{om}}(\mathrm{bd}(S_l)) < |\mathcal{L}|.
\]
Then we see that the boundary $\mathrm{bd}(S_l)$ is a Lebesgue null set. More specifically, a $\mathrm{cell}$ that is a $d$-dimensional $C^1$ submanifold  has locally finite Hausdorff measure.
Since $d < |\mathcal{L}|$, each cell has zero $|\mathcal{L}|$-dimensional Hausdorff measure~\citep[][Ch. 2.10]{federer2014geometric}. Hence, each cell is a  Lebesgue null set in  $\mathbb{R}^{|\mathcal{L}|}$~\citep[][Thm 2.10.35]{federer2014geometric}.
Hence, $\mathrm{bd}(S_l)$, as the union of finitely many cells, is a Lebesgue null set as required. 
\end{proof}

We provide now some remarks on Theorem \ref{thm:jacobian}. Item 1 in the theorem is the analog of the Williams–Daly–Zachary Theorem for discrete choice additive random utility models  (ARUM) \citep{mcfadden_econometric_1981,Sorensen2019}, which states that the gradient of the expected maximum utility is equal to the choice probability vector. The ARUM has an equivalent representation as a perturbed utility maximization problem similar to \eqref{eq:W}.

%Our analytical expression of the Jacobian $\nabla_c x^*(c)$ includes sensitivity information on PURC optimal flows with respect to changes of link costs in every direction, without the need to solve an auxiliary network flow problem. %This improves on previous results by  \cite{kyparisis1990solution,patriksson_sensitivity_2004} for the cases of directional derivatives. 

We note that the Jacobian $\nabla_c x^*(c)$ relates to all links in the network, including those with zero flow. Issues of non-differentiability occur only on the activation boundary where some links change between zero and positive flow. Using the convex conjugate of the perturbation function $F$, our theorem establishes continuous differentiability of the flow vector on the complement of the activation boundary. Moreover, we show that the activation boundary has Lebesgue measure zero, using again the properties of the convex conjugate. 

An alternative approach to establishing sensitivity results for the PURC would have been to use \citet[][Thm. 4]{tobin1988sensitivity} as a starting point. We have not pursued this approach, but under the strict complementarity and linear independence constraint qualification conditions assumed there, the equilibrium would lie in the complement of the activation boundary $\mathcal{C}_0$. In that case, the equilibrium conditions reduce to a smooth system and would have led us to the same Jacobian as in item 2 of Theorem \ref{thm:jacobian}. %, given that we would have been able to overcome additional mathematical challenges this approach would entail. 
In this sense, our activation boundary condition plays a role closely analogous to strict complementarity or non-degeneracy in the classical literature.

Our analysis, however, differs primarily in its formulation: rather than working with primal–dual optimality conditions as in the classic works, our theorem builds on the projected first-order condition. This approach avoids the introduction of dual variables and allows differentiability to be characterized directly in terms of the predicted flow, an object of primary interest.

%Differentiability of PURC flows $x^*$ almost everywhere follows from the specific structure of the model. Therefore, we do not have to impose general conditions such as the strict complementarity (SC) and linear independent constraint qualification (LICQ) in  and subsequent works. 

It is not a new insight that the set of points of non-differentiability is small. \citet[][Ch. 7.1]{boyles2020transportation} state (without proof) that there is a finite number of degenerate points in classical user equilibrium. 
Their argument is similar in spirit to our Lebesgue measure zero result, and our proof strategy could be applied in their case. 

The dual variables $\eta^*$ (node potentials) can be recovered from the first-order condition \eqref{eq:1}: for any active link $ij$ with $x^*_{ij} > 0$, we have $c_{ij} + F'_{ij}(x^*_{ij}) + \eta^*_j - \eta^*_i = 0$. Given the optimal flow $x^*$ and costs $c$, the dual variables are obtained by solving this linear system on the active subnetwork, after normalizing the potential at one node. The sensitivity of the (active) dual variables with respect to cost parameters then follows by differentiating \eqref{eq:1} with respect to $c$, using the Jacobian $\nabla_c x^*(c)$ from Theorem~\ref{thm:jacobian}.

In the next subsection, we will apply Theorem \ref{thm:jacobian} to sensitivity analysis of perturbed utility traffic equilibrium. Later, in Section~\ref{sec:numerical}, we apply Theorem \ref{thm:jacobian} to analyze the substitution pattern of the PURC, showing in particular that complementarity may occur. Something which is ruled out by ARUM discrete choice models.

\subsection{Flow-dependent case: stochastic traffic equilibrium}\label{sec:thm2}

In the previous section, we presented analytical sensitivity results for the PURC model for the case where link costs are independent of link flows. We now extend the analysis to the case of perturbed utility stochastic traffic equilibrium, where link costs are flow-dependent. 

We consider link cost functions parametrized by $\theta \in \Theta$ of arbitrary dimension,  denoted as $\zeta_{ij}(x_{ij}; \theta_{ij})$. The equilibrium link flows are parameterized by type-specific demands $q=(q^w)_{w \in \mathcal{W}}$, and we consequently denote the equilibrium link flows as $x^*_{ij}(c_{ij}^*; q)$ and the inverse link cost functions as $\zeta^{-1}_{ij}(c_{ij}; \theta_{ij})$. With this notation, the equilibrium condition \eqref{eq:prop_dual_fixed_point} becomes
\begin{equation}\label{eq:parametric_fixed_point}
   \zeta_{ij}^{-1}(c^*_{ij}; \theta_{ij}) = x^*_{ij}(c_{ij}^*; q), \forall ij \in \mathcal{L}.
\end{equation}

For notational ease, let $y = (\theta, q) \in \mathbb{R}^{|\mathcal{L}| + |\mathcal{W}|}$ denote the vector of parameters, with  
\begin{equation}
    G_{ij}(y, c) = \zeta_{ij}^{-1}(c^*_{ij}; \theta_{ij}) - x^*_{ij}(c_{ij}^*; q),
\end{equation} 
and vector-valued function $G(y, c) = (G_{ij}(y, c))_{ij \in\mathcal{L}}$.
We obtain the following result.
\begin{theorem}
    \label{thm:jacobian_equilibrium_cost}
    Suppose Assumptions \ref{A:1}-\ref{A:3} hold and that $U$ is an open subset of $\Theta \times \mathbb{R}^{|\mathcal{W}|}$ containing some parameter $y=(\theta, q) \in \Theta \times \mathbb{R}^{|\mathcal{W}|}$, such that there exists link costs $c^*$  satisfying the equilibrium condition $G(y, c^*) = 0$, and such that $\nabla_c x^{w*}(c^*)$ exists for each type $w\in\mathcal{W}$.  
    
    Then there exists a continuously differentiable function $c^*(y)$ for $y\in U$. Further, its Jacobian, $\nabla c^*(y)$  is
    \begin{align}
        \nabla c^*(y) = -\left[\nabla_{c} \zeta^{-1}(c^*; \theta) -  \nabla_c x^{*}(c^*; q)\right]^{-1} \nabla_{y} G(y, c^*), y \in U,\label{eq:jac_eq_cost}
    \end{align}
    where $$\nabla_y G(y, c^*)= \left[\nabla_\theta \zeta^{-1}(c^*; \theta) \middle| - \nabla_q x^*(c^*; q) \right] \in \mathbb{R}^{|\mathcal{L}| \times |y|} $$
    with $$\nabla_q x^*(c^*; q) = (x^{w*}(c^*))_{w \in \mathcal{W}}.$$
\end{theorem}
\begin{proof} By our assumption, $G(y, c)$ is differentiable at $(y, c^*)$. Differentiating $G(y, c^*)$ with respect to $c$ reads:
\begin{equation}\label{eq:equilibrium_jac_c}
    \nabla_c G(y, c^*) = \nabla_{c} \zeta^{-1}(c^*; \theta) -  \nabla_{c}x^{*}(c^*; q).
\end{equation}
By the implicit function theorem, the function $c^*: U \rightarrow \Theta \times \mathbb{R}^{|\mathcal{W}|}$ exists and is continuously differentiable if the Jacobian $\nabla_c G(y, c^*)$ is invertible.
$\nabla_{c} \zeta^{-1}(c^*; \theta)$ is a positive definite diagonal matrix by Assumption \ref{A:3}. 
    By Theorem~\ref{thm:jacobian}, we have that the Jacobian of PURC optimal link flows
    \begin{equation*}
        \nabla_{c} x^{w*}(c^*)=-\left( P^{w*}\nabla^2 F(x^{w*}(c^*))P^{w*}\right)^+,
    \end{equation*} is symmetric negative semi-definite.
    Hence, we have the Jacobian of the aggregated link flow 
    \begin{align*}
        \nabla_{c} x^*_{ij}(c^*; q) = \sum_{w\in\mathcal{W}} q^w  \nabla_{c}  x_{ij}^{w*}(c^*)
    \end{align*}
    is also symmetric negative semi-definite. Then $\nabla_c G(y, c^*)$ is positive definite and thus invertible. Consequently, the Jacobian of the equilibrium link costs function, $\nabla c^*(y)$ for $y \in U$ takes the following form:
    \begin{align}
        \nabla c^*(y) = - [\nabla_c G(y, c^*)]^{-1} \nabla_y G(y, c^*) = -\left[\nabla_{c} \zeta^{-1}(c^*; y) -  \nabla_c x^{*}(c^*; q)\right]^{-1} \nabla_y G(y, c^*),
    \end{align}
    where we have 
    $
\nabla_y G(y, c^*)= \left[\nabla_\theta \zeta^{-1}(c^*; \theta) \middle| - \nabla_q x^*(c^*; q) \right] \in \mathbb{R}^{|\mathcal{L}| \times |y|}$,
    with $\nabla_q x^*(c^*; q) = [\nabla^\top_q x_{ij}^*(c^*; q)]^\top_{ij \in \mathcal{L}}$ and $\nabla_q x_{ij}^*(c^*; q) = (x_{ij}^{w*}(c_{ij}^*))_{w \in \mathcal{W}}$.
\end{proof}

The proof of Theorem~\ref{thm:jacobian_equilibrium_cost} carries through even if the separability implicit in Assumption~\ref{A:3} is relaxed. Indeed, we only need link cost functions to be strictly monotone in aggregate link flows in the sense that $\nabla_{x} c(x^*; \theta)$ must be positive definite, which also ensures the existence of the inverse function $\zeta^{-1}(c; \theta)$ and positive definiteness of $\nabla_c \zeta^{-1}(c^*; \theta)$. 

Furthermore, Theorem \ref{thm:jacobian_equilibrium_cost} implies that equilibrium link costs are  functions of parameters $y$, that is $c^*(y)$.  Hence, the aggregate equilibrium link flows are also implicit functions of the parameters $y$, i.e., $x^*(c^*(y); q) = \sum_{q\in\mathcal{W}} q^w x^{w*}(c^*(y))$.
We can then compute the sensitivity of equilibrium link flows using the chain rule. We state this as a theorem for easy reference.
\begin{theorem}
   Suppose conditions in Theorem~\ref{thm:jacobian_equilibrium_cost} hold, the Jacobian of equilibrium link flows with respect to parameters $y$, $\nabla_y x^{w*}(c(y))$ and $\nabla x^*(y)$ respectively, are
    \begin{align}
        \nabla_y x^{w*}(c^*(y)) &= \nabla_c x^{w*}(c^*) \nabla_y c^*(y), \label{eq:type_specific_sensitivity}\\
        \nabla_y x^{*}(c^*(y)) &= \nabla_y x^*(c^*; y)+\sum_{w\in\mathcal{W}}q^w \nabla_y x^{w*}(c^*(y)),\label{eq:total_sensitivity}
    \end{align}
    where $\nabla_y x^*(c^*; y) = \left[\mathbf{0} \middle| \nabla_q x^*(c^*; q) \right] \in \mathbb{R}^{|\mathcal{L}|\times|y|}$.
\end{theorem}
\begin{proof}
    Under the conditions in Theorem~\ref{thm:jacobian_equilibrium_cost}, $x^{w*}(c^*(y))$ is differentiable at $y$, and so is $x^*(c^*(y); y)$. Therefore, $\nabla_y x^{w*}(c^*(y)) = \nabla_c x^{w*}(c^*) \nabla_y c^*(y)$, by the chain rule. Similarly,
    \begin{align*}
        \nabla_y x^*(c^*(y)) &= \nabla_y x^*(c^*; y) + \nabla_c x^*(c^*; y) \nabla_y c^*(y) \\
        &= \nabla_y x^*(c^*; y) + \sum_{w \in \mathcal{W}} q^w \nabla_c x^{w*}(c^*) \nabla_y c^*(y) \\
        &=\nabla_y x^*(c^*; y) + \sum_{w\in\mathcal{W}}q^w \nabla_y x^{w*}(c^*(y))
    \end{align*} and $\nabla_y x^*(c^*; y) = \left[\mathbf{0} \middle| \nabla_q x^*(c^*; q) \right]$.
\end{proof}

We note that, by exploiting the specific properties (strict convexity, separability, and differentiability, as detailed in Assumption~\ref{A:2}) of the perturbation function, Eq.~\eqref{eq:type_specific_sensitivity}-\eqref{eq:total_sensitivity} provide closed-form expressions for the Jacobian of the equilibrium link flows for all parameters $y$ at once. 
%This brings significant advantages compared to the directional derivative approaches of \cite{patriksson_sensitivity_2004}, \cite{josefsson2007sensitivity}, and \cite{lu2008sensitivity}, for which obtaining the full Jacobian requires solving the directional derivatives for each parameter of interest, one at a time.

\subsection{Welfare analysis of marginal changes}\label{sec:welfare}

We now discuss the application of item 1 in Theorem \ref{thm:jacobian} to welfare analysis. We have in mind situations where it is desired to assess the welfare consequences of an intervention in the transport system, which is perhaps relatively minor, such that the effort required to recompute the model for the case after the intervention might seem disproportionately large. In that case, it might be useful to be able to approximate the welfare change without recomputing the model. 

The value function \eqref{eq:W} measures the utility achieved by travelers of a certain type following a change in the link cost vector. Hence, the gradient of the value function in \eqref{eq:Roys} can be used to compute the approximate welfare change following a marginal change in link costs.

It is possible to apply a first-order Taylor approximation, whereby
    \[\Delta W(-c) \simeq x^*(c) \Delta c.\] 
This simply approximates the welfare change holding the flow vector constant.
    
The next step is to account also for the approximate change in the flow vector following a cost change. We can achieve this by a second-order Taylor approximation 
\begin{align}\label{eq:welfare_approx}
    \Delta W(-c) \simeq x^*(c) \Delta c + \frac 1 2 \Delta c^\top \nabla x^*(c) \Delta c.
\end{align}
This is reminiscent of the well-known rule-of-a half, commonly applied in benefit-cost analysis of transport system interventions \citep{fosgerau_rule---half_2021}.

The individual traveler's value function \eqref{eq:W} depends only on the cost vector, regardless of whether it is flow-dependent or not. 
For the flow-dependent case of equilibrium described in Section \ref{sec:thm2}, we can use the chain rule, combining \eqref{eq:Roys} with \eqref{eq:jac_eq_cost} in Theorem \ref{thm:jacobian_equilibrium_cost}, to estimate welfare changes due to changes in network parameters.

\subsection{Algorithms for computing sensitivity}

\subsubsection{Computation of sensitivity of PURC} 

The primary challenge in computing the sensitivity of the optimal PURC link flows is evaluating the Jacobian in Eq.~\eqref{eq:eqngradient}. 
The naive approach, computing the pseudo-inverse of a matrix the size of the full network $|\mathcal{L}|\times |\mathcal{L}|$, is computationally prohibitive, with a complexity of $O(|\mathcal{L}|^3)$. We propose instead to exploit the inherent sparsity of the problem to find a much more efficient approach. 

For any given traveler type $w$, the optimal flow vector $x^{w*}$ is typically sparse,  as many links on irrelevant parts of the network have zero flow. 
The corresponding entries in the Jacobian matrix for these inactive links are simply zero. This allows us to restrict the expensive matrix operations to the much smaller active subnetwork where flows are strictly positive, $x_{ij}^{w*} > 0$.
 
The streamlined computational procedure is as follows. For each type $w \in \mathcal{W}$, we first identify the active subnetwork, which consists only of the set of active links $\mathcal{L}^{w*}$ and their incident nodes. We then construct a node-link incidence matrix for this active subnetwork and remove the row corresponding to the destination of $w$, denoting the resulting matrix as $\tilde{A}^{w*}$. This matrix has full row rank~\citep[see e.g., Prop. 4.3 in][]{biggs1993algebraic},  which makes the term $\left(\tilde{A}^{w*}\right)^\top \tilde{A}^{w*}$ invertible.
We use $\sim$ to denote quantities related to the active subnetwork. The key algorithmic steps are then summarized in Algorithm~\ref{algo}.

\begin{algorithm}[H] \label{algo}
\DontPrintSemicolon
\KwIn{PURC optimal link flows $(x^{w*})_{w\in\mathcal{W}}$}

\For{$w \in\mathcal{W}$}{
    \textbf{1. Construct the projection matrix $\tilde{P}^{w*}$}
    
    The projection matrix for the active subnetwork is
    \begin{equation}\label{eq:sparse_proj}
        \tilde{P}^{w*} = \mathbb{I} - \left(\tilde{A}^{w*}\right)^\top \left[\left(\tilde{A}^{w*}\right)^\top \tilde{A}^{w*}\right]^{-1} \tilde{A}^{w*}
    \end{equation}

    \textbf{2. Construct the Jacobian of the PURC optimal link flow}
    \begin{equation}
        \nabla_{\tilde{c}} \tilde{x}^{w*}(c) = \left[\tilde{P}^{w*} \nabla_{\tilde{x}}^2 F(x^{w*}) \tilde{P}^{w*} \right]^+,
    \end{equation}
    where $\nabla_{\tilde{x}}^2 F(x^{w*}) = \text{diag}(\tilde{h}^{w*})$ with $\tilde{h}^{w*} = (\nabla^2 F_{ij}(x_{ij}^{w*}))_{ij \in \mathcal{L}^{w*}}$. Then, the full Jacobian matrix is:
    \begin{equation}
        \nabla {x}^{w*}(c) = \begin{pmatrix}
\nabla_{\tilde{c}} \tilde{x}^{w*}(c) & \mathbf{0} \\ 
\mathbf{0} & \mathbf{0}%
\end{pmatrix},
    \end{equation}
    up to re-arrangement for indexing consistency.
}

\caption{Algorithm for computing PURC sensitivity}
\end{algorithm}
%We note that Eq.~\eqref{eq:sparse_proj} has been applied in our ongoing work on developing a regression-based estimator for the PURC model that is applicable to individual-level data. 
This approach has significant computational benefits. It dramatically reduces the computational complexity from $O(|\mathcal{L}|^3)$ to $O(|\tilde{\mathcal{L}}^{w*}|^3)$, where $|\tilde{\mathcal{L}}^{w*}|$ is typically far smaller than $|\mathcal{L}|$. This makes the sensitivity analysis scalable to very large, real-world networks, as will be shown in our large-scale experiment below. 

The computations in Algorithm \ref{algo} are carried out independently for each travel type $w \in \mathcal{W}$. This makes the algorithm  embarrassingly parallelizable by traveler type, which allows for substantial speedups with multi-core processing units. 

The sensitivity of the flow for each traveler type is an input for the computation of the sensitivity of traffic equilibrium, discussed in the next section.

\subsubsection{Computation of sensitivity for traffic equilibrium}
 
The sensitivity analysis of traffic equilibrium requires evaluating the Jacobian from Eq.~\eqref{eq:jac_eq_cost}, which involves inverting a large, network-sized matrix $\nabla_{c} \zeta^{-1}(c^*; \theta) -  \nabla_c x^{*}(c^*; q)$, where  the aggregate Jacobian $\nabla_c x^{*}(c^*; q)$ is assembled from the per-type Jacobians computed by Algorithm \ref{algo}, evaluated at the equilibrium costs $c^*
$. Direct inversion of this matrix presents a significant computational bottleneck for large-scale networks.
To overcome this, we propose an efficient two-step approach that exploits the inherent structure of the equilibrium problem. The core idea is to first reduce the problem's dimensionality by leveraging sparsity and then to solve the reduced problems using a numerically superior method.

First, we exploit again the fact that, in equilibrium, many links in the network are inactive, meaning they carry zero aggregate flows $x^*_{ij} = 0$. By partitioning the links into active $\tilde{\mathcal{L}^*} = \{ij| x^*_{ij} = \sum_{w\in\mathcal{W}}q^w x_{ij}^{w*}>0\}$ and inactive sets, we can reorder the matrix into a simpler block-diagonal structure. This isolates the computationally intensive part of the problem into two smaller sub-matrices: 
\begin{equation}
    \nabla_{c} \zeta^{-1}(c^*; \theta) -  \nabla_c x^{*}(c^*; q) = \begin{pmatrix}
D_1 - \nabla_{\tilde{c}} \tilde{x}^{*}(c)  & \mathbf{0} \\ 
\mathbf{0} & D_2%
\end{pmatrix},
\end{equation}
where the upper-left block corresponds to the active subnetwork.  $D_1$ is a diagonal matrix containing the inverse link cost functions, $\zeta_{ij}^{-1}(c^*_{ij}; \theta)$, for the active links $ij \in \tilde{\mathcal{L}^*}$. Similarly, $\nabla_{\tilde{c}} \tilde{x}^{*}(c)$ is the Jacobian of the active aggregated PURC flows with respect to the costs on the active links, that is $\nabla_{\tilde{c}} \tilde{x}^{*}(c) = \left(\sum_{w\in\mathcal{W}} q^w \nabla^\top_{\tilde{c}}{x}_{ij}^{w*}(c)\right)^\top_{ij\in\tilde{\mathcal{L}}^*}$ with $\tilde{c} = (c_{ij})_{ij \in \tilde{\mathcal{L}}^*}$. The bottom-right block corresponds to the inactive links. Since the Jacobian of PURC flows is zero for inactive links, this block simplifies to $D_2$, a diagonal matrix of the inverse link cost derivatives for inactive links, $ij \in \mathcal{L} \setminus \tilde{\mathcal{L}}^*$.
Consequently, the inverse of the diagonal block $D_2$ is trivial. The only remaining challenge is to handle the smaller, dense block corresponding to the active subnetwork. 

Second, we consider the inversion of the matrix $M \coloneqq D_1 - \nabla_{\tilde{c}} \tilde{x}^{*}(c)$ for the  active subnetwork block. Computing the inverse $M^{-1}$ first is not the most efficient or stable method.  
Instead, we may observe that our goal is to compute $J \coloneqq M^{-1}\nabla_{\tilde{y}} \tilde{G}(y, c^*)$. We propose to employ an efficient and numerically stable computational routine~\citep[see e.g., Ch. 3 in][]{golub2013matrix} for that purpose: finding $J$ by solving the equivalent linear system of equations $MJ=\nabla_{\tilde{y}} \tilde{G}(y, c^*)$.  This method allows leveraging highly optimized and numerically stable algorithms to find $J$ without ever needing to form the explicit inverse. For medium-size matrices (e.g., up to $10,000\times10,000$ on modern machines), direct solution routines like Cholesky or LDLT factorization \citep{golub2013matrix} are highly efficient with time complexity of ${|\tilde{\mathcal{L}}^*|}^3/3 + O(|\tilde{\mathcal{L}}^*|^2)$, where $|\tilde{\mathcal{L}}^*|$ is the number of active links. For extremely large matrices, Krylov subspace methods like the conjugate gradient method \citep{golub2013matrix} are well suited, whose complexity for large linear systems can be magnitudes smaller than direct methods~\citep[see e.g., Thm 6.29 in][]{saad2003iterative}.   

Combining these two steps, the block-diagonal transformation to reduce the problem's size and the solution of a linear system as the core computational step, we obtain an algorithm that is scalable, numerically stable, and practical for implementation in large, real-world networks.

The computational approach exploits sparsity at two distinct levels. At the first level, Algorithm \ref{algo} computes the per-type Jacobians $\nabla_c x^{w*}(c^*)$ on each type's active subnetwork $\tilde{\mathcal{L}}^{w*}$, which is typically very small relative to the full network. This is where the most substantial computational savings occur. At the second level, the block-diagonal reduction in the equilibrium sensitivity computation exploits sparsity in the aggregate active subnetwork $\tilde{\mathcal{L}}^* = \{ij : \sum_w q^w x_{ij}^{w*} > 0\}$. When the number of OD pairs is large and traffic analysis zones are fine, most links may carry positive aggregate flow, limiting the compression at this second level. In such cases, the primary efficiency gain would come from Algorithm \ref{algo} together with the fact that the equilibrium sensitivity is obtained by solving a linear system rather than recomputing the full equilibrium.

\section{Numerical examples}\label{sec:numerical}

In this section, we show several applications of the sensitivity analysis of traffic equilibrium using an example network. Specifically, we will show how the sensitivity information can be used for estimating optimal solutions subject to a shift in the design parameters, something which is often required for bilevel optimization problems. We will show how sensitivity information can assist in identifying important design parameters, for which small changes will lead to large changes in network performance. We will apply the sensitivity-based uncertainty analysis to quantify the impacts of model parameter uncertainty on uncertainty in performance predictions. Lastly, we use sensitivity analysis of PURC to illustrate that the PURC model can capture complementarity between paths. The latter is a counterexample to  Proposition 3 in \citet{Fosgerau2021}, and we point out the error in their proof.

\begin{figure}[H]
    \centering
\includegraphics[width=0.5\linewidth]{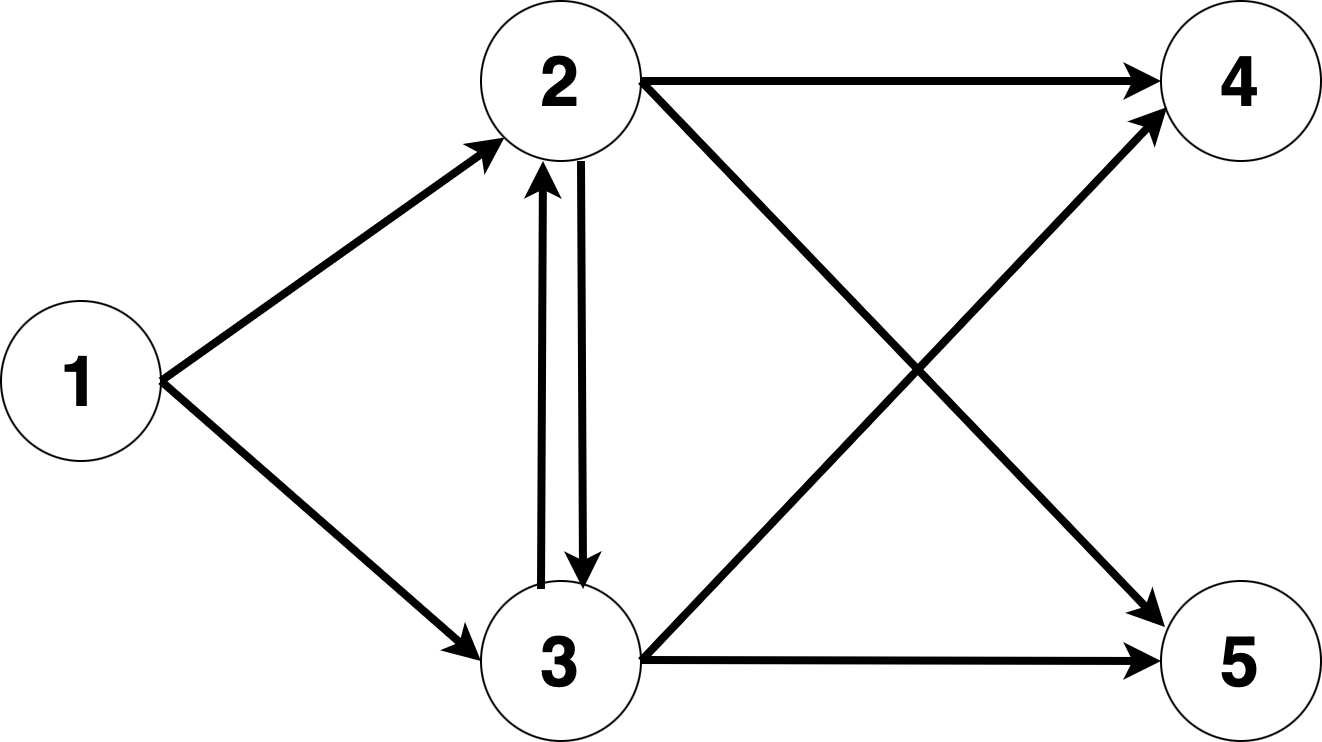}
    \caption{Example network}\label{fig:example_network}
\end{figure}

We consider the example network in Figure~\ref{fig:example_network}, which has  5 nodes and 8 directional links, as well as two origin-destination (OD) pairs $(1, 4)$, $(1, 5)$ with demands of 15 and 20, respectively. We assume the link costs to follow BPR link cost-time functions 
\begin{equation} \label{bpr}
    \zeta_{ij}(x_{ij}) = t_{0,ij}\left(1+\alpha\left(\frac{x_{ij}}{\kappa_{ij}}\right)^\beta\right),
\end{equation}
where $t_{0,ij}$ are free-flow travel times, $\alpha=0.15$, $\beta=4$, and $\kappa_{ij}$ are link capacities, with numerical values specified in Table~\ref{tab:brp_param}. Unless otherwise specified, the perturbation function in this section is the entropy-based Eq.~\eqref{eq:Fentropic}.
\begin{table}[H]
\captionsetup{justification=centering}
\caption{Link performance parameters.}\label{tab:brp_param}
\setlength\tabcolsep{8pt}
\centering
\begin{tabular}{ccccccccc}
                         & $(1, 2)$  & $(1, 3)$  & $(2, 3)$ &  $(2, 4)$ & $(2, 5)$ & $(3, 2)$ & $(3, 4)$ & $(3, 5)$\\
                        \toprule
$t_{0,ij}$              & 3.0 & 5.0 & 1.0 & 5.0 & 6.0 & 1.0 & 4.0 & 4.0\\
\midrule
$\kappa_{ij}$       & 30.0 & 30.0 & 15.0 & 15.0 & 15.0 & 15.0 & 15.0 & 15.0\\
                        \bottomrule
\end{tabular}
\end{table}
Table~\ref{tab:exact_estimate_opt} reports the equilibrium link flows at these parameters. The equilibrium is computed using the algorithm presented in \cite{yao_perturbed_2024}.

\subsection{Sensitivity analysis of equilibrium}

Table~\ref{tab:jac_x_kappa} shows the sensitivity of equilibrium link flows $x^*$ with respect to link capacities $\kappa$ and Table~\ref{tab:jac_x_t0} shows the same with respect to free-flow travel times $t_{0}$. Both tables are calculated using Theorem \ref{thm:jacobian_equilibrium_cost}. We first note that the signs in the Jacobian are immediately plausible given the network structure. We also observe as a check that, for any $\kappa_{ij}$, $\forall ij \in \mathcal{L}$ (and similarly for $t_{0,ij}$), $\partial x^*_{25}(\kappa)/\partial \kappa_{ij} + \partial x^*_{35}(\kappa)/\partial \kappa_{ij} = \partial x^*_{24}(\kappa)/\partial \kappa_{ij} + \partial x^*_{34}(\kappa)/\partial \kappa_{ij} = 0$, which implies, as expected, that flow conservation holds. 
\begin{table}[H]
\captionsetup{justification=centering}
\caption{Jacobian of equilibrium link flow with respect to link capacity.}\label{tab:jac_x_kappa}
\setlength\tabcolsep{8pt}
\centering
\begin{tabular}{crrrrrrrr}
 $\nabla x^*(\kappa)$ & $\kappa_{12}$ &  $\kappa_{13}$      & $\kappa_{23}$      & $\kappa_{24}$ &     $\kappa_{25}$      & $\kappa_{32}$     & $\kappa_{34}$      & $\kappa_{35}$      \\
\toprule
$x^*_{12}$                         & 0.356                 & -0.004 & 0.083  & 0.046  & 0.010  & 0.000 & -0.006 & -0.128 \\
$x^*_{13}$                         & -0.356                & 0.004  & -0.083 & -0.046 & -0.010 & 0.000 & 0.006  & 0.128  \\
$x^*_{2 3}$                         & 0.245                 & -0.003 & 0.164  & -0.096 & -0.022 & 0.000 & 0.012  & 0.286  \\
$x^*_{24}$                         & 0.064                 & -0.001 & -0.045 & 0.150  & -0.002 & 0.000 & -0.018 & 0.030  \\
$x^*_{25}$                         & 0.048                 & -0.001 & -0.036 & -0.008 & 0.034  & 0.000 & 0.001  & -0.444 \\
$x^*_{32}$                         & 0.000                 & 0.000  & 0.000  & 0.000  & 0.000  & 0.000 & 0.000  & 0.000  \\
$x^*_{3 4}$                         & -0.064                & 0.001  & 0.045  & -0.150 & 0.002  & 0.000 & 0.018  & -0.030 \\
$x^*_{35}$                         & -0.048                & 0.001  & 0.036  & 0.008  & -0.034 & 0.000 & -0.001 & 0.444 \\
\bottomrule
\end{tabular}
\end{table}

\begin{table}[H]
\captionsetup{justification=centering}
\caption{Jacobian of equilibrium link flow with respect to free-flow travel time.}\label{tab:jac_x_t0}
\setlength\tabcolsep{8pt}
\centering
\begin{tabular}{crrrrrrrr}
 $\nabla x^*(t_0)$ & $t_{0,12}$ &  $t_{0,13}$      & $t_{0,23}$      & $t_{0,24}$ &     $t_{0,25}$      & $t_{0,32}$     & $t_{0,34}$      & $t_{0,35}$      \\
\toprule
$x^*_{12}$                         & -9.775                 & 8.891 & -6.405  & -1.626  & -1.204  & 0.000 & 1.597 & 1.317 \\
$x^*_{13}$                         & 9.775                & -8.891  & 6.405 & 1.626 & 1.204 & 0.000 & -1.597  & -1.317  \\
$x^*_{2 3}$                         & -6.706                 & 6.099 & -12.731 & 3.406 & 2.700 & 0.000 & -3.345  & -2.954  \\
$x^*_{24}$                         & -1.751                 & 1.593 & 3.503 & -5.319  & 0.283 & 0.000 & 5.224 & -0.309  \\
$x^*_{25}$                         & -1.318                 & 1.198 & 2.823 & 0.287 & -4.186  & 0.000 & -0.282  & 4.581 \\
$x^*_{32}$                         & 0.000                 & 0.000  & 0.000  & 0.000  & 0.000  & 0.000 & 0.000  & 0.000  \\
$x^*_{3 4}$                         & 1.751                & -1.593  & -3.503  & 5.319 & -0.283  & 0.000 & -5.224  & 0.309 \\
$x^*_{35}$                         & 1.318                & -1.198  & -2.823  & -0.287  & 4.186 & 0.000 & 0.282 & -4.581  \\
\bottomrule
\end{tabular}
\end{table}

% In addition, we see that, compared to link capacities, changes in free-flow travel times are estimated to have greater impacts (in terms of the sum of absolute gradient values) on equilibrium link flows, even though link capacities are component in the exponent of the link cost functions. In particular, increasing free-flow travel times (or, reducing the maximum speed) on link $(2, 3)$ result in greatest variation in equilibrium link flows, compared to changing other model parameters. For link capacities, increasing capacity on link $(3, 5)$ has higher impact on equilibrium link flows. On the other hand, centroid connector links $(1, 2)$ and $(1, 3)$ are more sensitive to changes in link capacities in the network (in terms of sum of absolute partial derivatives); while link $(2, 3)$ is more sensitivity to changes in free-flow travel times. 

\subsection{Solution estimation}
A classical application of sensitivity analysis is to estimate the new optimal solution after a shift in the design parameters. This is useful since it can be much faster to compute the approximate new solution than to solve the model exactly. Such computation efficiency is critical for bilevel optimization problems~\citep[e.g.,][]{yang1997traffic,josefsson2007sensitivity,liu2022inducing}.

For illustration, we consider two scenarios: i) increase the capacity on link $(1, 2)$ by $5\%$; ii) increase the free-flow travel time on link $(1, 2)$ by $5\%$. Our goal is to compare the optimal link flows computed exactly, rerunning the traffic equilibrium problem, and approximated using the first-order Taylor series approximation
\begin{equation}\label{taylor}
    x^{*}(\theta +\epsilon_{\theta}) \approx x^{*}(\theta) +  \nabla x^*(\theta)\epsilon_{\theta},
\end{equation}
where $x^{*}(\theta +\epsilon_{\theta})$ denotes the optimal link flows subject to a shift $\epsilon_{\theta}$ in parameters $\theta$, and $\nabla x^*(\theta)$ is computed using Eq.~\eqref{eq:total_sensitivity}. The results in Table~\ref{tab:exact_estimate_opt} show that the approximate solutions are close to the exact solutions after a 5\% shift in the two model parameters.

\begin{table}[H]
\captionsetup{justification=centering}
\caption{Estimated and exact solutions under shifts $\epsilon_{\theta}$.}\label{tab:exact_estimate_opt}
\setlength\tabcolsep{6pt}
\centering
\begin{tabular}{c|r|rrr|rrr}
\multicolumn{1}{c}{\multirow{2}{*}{\textbf{Link}}} & \multicolumn{1}{c}{\multirow{2}{*}{\textbf{Baseline}}} & \multicolumn{3}{c}{\textbf{Increase $\kappa_{(1,2)}$ by 5\%}} & \multicolumn{3}{c}{\textbf{Increase $t_{0,(1,2)}$ by 5\%}} \\
\multicolumn{1}{c}{}                       & \multicolumn{1}{c}{}                             & Exact           & Approximate         & Difference         & Exact             & Approximate            & Difference           \\
\midrule
$(1, 2)$                                          & 27.127                                           & 27.631          & 27.662            & -0.031             & 25.633            & 25.661               & -0.028               \\
$(1, 3)$                                          & 7.873                                            & 7.370           & 7.339             & 0.031              & 9.367             & 9.339                & 0.027                \\
$(2, 3)$                                          & 11.446                                           & 11.790          & 11.813            & -0.023             & 10.405            & 10.440               & -0.035               \\
$(2, 4)$                                          & 9.233                                            & 9.324           & 9.329             & -0.005             & 8.973             & 8.971                & 0.002                \\
$(2, 5)$                                         & 6.448                                            & 6.517           & 6.520             & -0.003             & 6.255             & 6.250                & 0.005                \\
$(3, 2)$                                          & 0.000                                            & 0.000           & 0.000             & 0.000              & 0.000             & 0.000                & 0.000                \\
$(3, 4)$                                         & 5.767                                            & 5.676           & 5.671             & 0.005              & 6.027             & 6.030                & -0.003               \\
$(3, 5)$                                          & 13.552                                           & 13.483          & 13.480            & 0.003              & 13.744            & 13.750               & -0.006              \\
\bottomrule
\end{tabular}
\end{table}

\subsection{Analysis of the propagation of uncertainty}
When model parameters are estimated with uncertainty, model predictions will also be uncertain.
Sensitivity-based uncertainty analysis concerns how uncertainty in model parameters propagates to model predictions in terms of link flow patterns at equilibrium and to determine the confidence level of the model predictions. We adopt the classic approach~\citep[e.g., applied in][]{yang2013sensitivity, du2022sensitivity} for uncertainty analysis. Specifically, let us suppose now that the model parameters are a random vector $\Theta$ with mean $\mu_{\Theta}$ and covariance matrix $K_{\Theta}$,
the expectation and covariance of the equilibrium link flows $X^*=x^*(\Theta)$ are approximated by:
\begin{align}
    \mathbb{E}[X^*] &= \mathbb{E}[x^*(\Theta)]\approx x^*(\mu_{\Theta}) + \mathbb{E}[(\Theta - \mu_{\Theta})] [\nabla x^*(\mu_{\Theta})]^\top = x^*(\mu_{\Theta}) \label{eq:mean_est}\\
    \text{Var}[X^*] &= \mathbb{E}[(X^* - \mathbb{E}[X^*])(X^* - \mathbb{E}[X^*])^\top] \nonumber \\
    &\approx \mathbb{E}[(x^*(\Theta) - x^*(\mu_{\Theta}))(x^*(\Theta) - x^*(\mu_{\Theta}))^\top] \nonumber \\
    &\approx \mathbb{E}[((\Theta - \mu_{\Theta})[\nabla x^*(\mu_{\Theta})]^\top )([\nabla x^*(\mu_{\Theta})])(\Theta - \mu_{\Theta})^\top)] \nonumber \\
    & = \nabla x^*(\mu_{\Theta}) K_\Theta [\nabla x^*(\mu_{\Theta})]^\top. \label{eq:uncertainty_cov}
\end{align}
We can then derive the covariance and correlation between equilibrium link flow $X^*$ and model parameters $\Theta$ by the multivariate delta method:
\begin{align}
 \text{Cov}[X^*, \Theta] &\approx \nabla x^*(\mu_{\Theta}) K_\Theta \\
 \text{Corr}[X^*, \Theta] &= [\text{diag(}\sigma_{X^*})]^{-1} \text{Cov}[X^*, \Theta] [\text{diag(}\sigma_{\Theta})]^{-1}\label{eq:corr_est},
\end{align}
where $\text{diag}(\sigma_{X^*}), \text{diag(}\sigma_{\Theta})$ are diagonal matrix whose diagonal elements are standard deviations of $X^*$ and $\Theta$, respectively. This covariance and correlation information will provide insights into the importance of model parameters on model predictions, such as for selecting critical parameters.

For demonstration, we consider a similar experiment as in~\cite{yang2013sensitivity}. We assume that link capacities are distributed with mean $\mu_{\kappa}$ as specified in Table~\ref{tab:brp_param} and coefficient of variation (CV) of 0.20. Table~\ref{tab:uncertainty} reports i) the mean equilibrium link flows,  computed using Eq.~\eqref{eq:mean_est}, ii) the standard deviation, which is the square root of the diagonal entries of Eq.~\eqref{eq:uncertainty_cov}, and iii) the coefficient of variation, computed using the estimated mean and standard deviation. 

\begin{table}[H]
\captionsetup{justification=centering}
\caption{Uncertainty of equilibrium link flows due to uncertainty in link capacities.}\label{tab:uncertainty}
\setlength\tabcolsep{8pt}
\centering
\begin{tabular}{crrr}
\multicolumn{1}{c}{{\textbf{Links}}} & \multicolumn{1}{c}{{\textbf{Mean}}} & {{\textbf{Std. Dev.}}} & {{\textbf{CV}}} \\
\midrule
$(1, 2)$                                          & 27.127                                           & 2.191          & 0.079                        \\
$(1, 3)$                                          & 7.873                                            & 2.191           & \textbf{0.299}             \\
$(2, 3)$                                          & 11.446                                           & 1.795          & 0.152                \\
$(2, 4)$                                          & 9.233                                            & 0.614           & 0.066                         \\
$(2, 5)$                                         & 6.448                                            & 1.371           & \textbf{0.210}                         \\
$(3, 2)$                                          & 0.000                                            & 0.000           & -                   \\
$(3, 4)$                                         & 5.767                                            & 0.614           & 0.108                      \\
$(3, 5)$                                          & 13.552                                           & 1.371         & 0.101                  \\
\bottomrule
\end{tabular}
\end{table}

We observe in Table~\ref{tab:uncertainty} that the CV of most links is smaller than that of the link capacities (0.20), while uncertainties on links $(1, 3), (2, 5)$ are relatively larger, potentially due to their larger free-flow travel times that scales up the variations. 

\begin{table}[H]
\captionsetup{justification=centering}
\caption{Correlation between equilibrium link flow and link capacity.}\label{tab:corr}
\setlength\tabcolsep{8pt}
\centering
\begin{tabular}{crrrrrrrr}
 $\text{Corr}[X^*, \kappa]$ & $\kappa_{12}$ &  $\kappa_{13}$      & $\kappa_{23}$      & $\kappa_{24}$ &     $\kappa_{25}$      & $\kappa_{32}$     & $\kappa_{34}$      & $\kappa_{35}$      \\
\toprule
$x^*_{12}$                         & 0.976                 & -0.012 & 0.113  & 0.063  & 0.013  & 0.000 & -0.008 & -0.175 \\
$x^*_{13}$                         & -0.976                & 0.012  & -0.113 & -0.063 & -0.013 & 0.000 & 0.008  & 0.175  \\
$x^*_{2 3}$                         & 0.817                 & -0.010 & 0.275  & -0.160 & -0.037 & 0.000 & 0.019  & 0.479  \\
$x^*_{24}$                         & 0.624                 & -0.007 & -0.221 & 0.730  & -0.011 & 0.000 & -0.089 & 0.146  \\
$x^*_{25}$                          & 0.210                 & -0.002 & -0.080 & -0.018 & 0.075  & 0.000 & 0.002  & -0.971 \\
$x^*_{32}$                         & 0.000                 & 0.000  & 0.000  & 0.000  & 0.000  & 0.000 & 0.000  & 0.000  \\
$x^*_{3 4}$                         & -0.624                & 0.007  & 0.221  & -0.730 & 0.011  & 0.000 & 0.089  & -0.146 \\
$x^*_{35}$                         & -0.210                & 0.002  & 0.080  & 0.018  & -0.075 & 0.000 & -0.002 & 0.971  \\
\bottomrule
\end{tabular}
\end{table}

In addition, we can estimate the correlation between equilibrium link flows and link capacities using Eq.~\eqref{eq:corr_est}. As shown in Table~\ref{tab:corr}, link capacities are mostly correlated with the corresponding equilibrium link flows, i.e., diagonal values are larger than off-diagonal values in each column. We also observe that the capacity on link $(1, 2)$ has strong correlations with equilibrium link flows, compared to link capacities on other links. This implies that $\kappa_{12}$ could have significant impacts on equilibrium link flow patterns in the network.

\subsection{Substitution patterns }\label{sec:complementarity}

The PURC model does not have error terms in the way that a classical discrete choice model in its random utility representation has random utility components. Therefore, talking about correlation in the PURC model is not immediately meaningful. Instead, we discuss the model in terms of substitution and complementarity, which is what we are concerned about. 
Specifically, we say that substitution occurs when a cost increase for an alternative causes the demand to increase for another alternative. In contrast, complementarity occurs when a cost increase for an alternative causes the demand to decrease for some other alternative. As we will show in this section, the PURC model allows paths to be complements. This is in contrast to the ARUM discrete choice model, in which alternatives can only be substitutes \citep{Fosgerau2013y}. 

It is clear that both substitutability and complementarity are at work when considering route choice at the level of network links: a cost increase on a link tends to reduce flow on nearby upstream and downstream links, while flow tends to increase on links on alternative paths. The situation is different when we consider substitution over paths. The path alternatives are necessarily substitutes in a classical discrete choice model, an additive random utility discrete choice model over paths. This fact follows from the properties of general additive random utility models \citep[e.g.,][]{Fosgerau2013y}. In this paper, by applying our sensitivity analysis results, we will show via an example that complementarity among paths is possible in the PURC model in general.

To illustrate the substitution pattern of PURC, we consider a simple example network (Figure \ref{fig:1_complementarity_example}) with equal link lengths of 1,  (flow-independent) link costs equal to link lengths, and a single OD pair $(1, 3)$. Suppose a positive flow vector $x^*$ such that all links are active, and that the perturbation function $F(x)$ takes the form of Eq.~\eqref{eq:Fquad}.

Applying Theorem \ref{thm:jacobian}, the Jacobian of $x^*$ with respect to the cost vector can be found to be the following matrix. 

\[
 \nabla_c x^* = \begin{bmatrix}

 &(12)&(15)&(23)&(24)&(43)&(54)&(53) \\
(12)&-0.333&0.333&-0.167&-0.167&0&0.167 &	0.167 \\
(15)&0.333 &-0.3330& 0.167&0.167&0&-0.167&-0.167\\
(23)&-0.167&0.167&-0.458&0.292&0.25&\mathbf{\tr -0.0417}&0.208\\
(24)&-0.167&0.167&0.292&-0.458&-0.25&0.208&-0.0417\\
(43)&0&0&0.25&-0.25&-0.5&-0.25&0.25\\
(54)&0.167&-0.167&\mathbf{\tr -0.0417}&0.208&-0.25&-0.458&0.292\\
(53)&0.167&-0.167&0.208&-0.0417&0.25&0.292&-0.458
\end{bmatrix}
\]

By construction, this example network contains 4 paths, where path flows are equal to the flows on links $(23)$, $(24)$, $(54)$, and $(53)$, respectively.
Consider the entries marked in bold, corresponding to $\frac{\partial x_{54}}{\partial c_{23}} $ and by symmetry to $\frac{\partial x_{23}}{\partial c_{54}}$. This is negative, which shows that paths $1\rightarrow{}2\rightarrow{}3$ and $1\rightarrow{}5\rightarrow{}4\rightarrow{}3$ are complements and not substitutes: the flow on the latter route decreases when the cost on link $23$ on the first route increases.

We take this opportunity to point out that this contradicts Proposition 3 in \citet{fosgerau_perturbed_2022}, which erroneously stated that paths are necessarily substitutes in PURC.
In the proof of that proposition, we rewrote the perturbation function $F$ defined in terms of link flows as a function $G$ defined in terms of path probabilities. We then claimed that $G$ is strictly convex. However, that is not generally true, since the decomposition of a flow into path flows is not unique.

\begin{figure}[H]
  \centering
\begin{tikzpicture}[->,shorten >=1pt,auto,node distance=2cm,thick]
  % Define nodes
\node[draw=black,circle,label=above:{$\text{Origin}$}] (1) at (0,0) {$1$};
  \node[draw=black,circle] (2) at (-2,-2) {$2$};
  \node[draw=black,circle, label=below:{$\text{Destination}$}] (3) at (0,-5) {$3$};
  \node[draw=black,circle] (4) at (0,-3) {$4$};
  \node[draw=black,circle] (5) at (2,-2) {$5$};
  
  % Draw directed edges with labels
  \path[every node/.style={font=\sffamily\small}]
    (1) edge node[above] {1} (2)
    (2) edge node[below] {1} (3)
    (2) edge node[above] {1} (4)
    (4) edge node[right] {1} (3)
    (1) edge node[above] {1} (5)
    (5) edge node[above] {1} (4)
    (5) edge node[below] {1} (3);
\end{tikzpicture}
\caption{Simple network example, substitution pattern of PURC.}\label{fig:1_complementarity_example}
\end{figure}

\section{Large-scale demonstration}\label{sec:large_scale}

In this section, we demonstrate the application of flow-dependent sensitivity analysis of the PURC model in a large-scale case, covering the Copenhagen Metropolitan Area using a network containing 30,773 links and 12,876 nodes. Each link $a$ has a BPR cost-time function ~\eqref{bpr}, with specific parameters. In \cite{fosgerau_perturbed_2022}, we utilized a GPS dataset with 1,234,289 observed route choices to estimate the PURC model. We apply their Model C, which includes pace parameters (min/km) specific to link-types, a parameter for a dummy if links lead into intersections (i.e. the number of outlinks excluding u-turns is at least two), and uses the entropy link perturbation function from Eq.~\eqref{eq:Fentropic}. The parameter estimates are shown in Table \ref{tab:copenhagen}.

\begin{table}[H]
\captionsetup{justification=centering}
\caption{Estimated cost parameters for the Copenhagen area~\citep{fosgerau_perturbed_2022}.}\label{tab:copenhagen}
\setlength\tabcolsep{5pt}
\centering
\begin{tabular}{l|llllllll}
   & \multicolumn{8}{c}{\textbf{Pace} {[}min/km{]}}                                                                                                                                                                                                                                                                                                                                                                                                            \\
 \textbf{Dummy}                           & Motorways & \begin{tabular}[c]{@{}l@{}}Motorway \\ ramps\end{tabular} & \begin{tabular}[c]{@{}l@{}}Motor \\ traffic\\  roads\end{tabular} & \begin{tabular}[c]{@{}l@{}}Other \\ national \\ roads\end{tabular} & \begin{tabular}[c]{@{}l@{}}Urban \\ roads\end{tabular} & \begin{tabular}[c]{@{}l@{}}Rural \\ roads\end{tabular} & \begin{tabular}[c]{@{}l@{}}Smaller \\ roads\end{tabular} & \begin{tabular}[c]{@{}l@{}}Other \\ ramps\end{tabular} \\
                            \midrule
0.0155 & 0.3501   & 0.5542                                                   & 0.5623                                                           & 0.5997                                                            & 0.5692                                               & 0.7778                                               & 0.5300                                                 & 0.4456  \\
\bottomrule
\end{tabular}
\end{table}

%begin{figure}[H]
%\centering
%\subfloat[\label{fig:totalFlowDiff}]
%{\includegraphics[width=0.49\textwidth]{SubstitutionPaperOverallFlow_5sec kopiér.png}} %\hfill
%\subfloat[\label{fig:ODflowshareChange}]{\includegraphics[width=0.49\textwidth]{SubstitutionPaperOverallFlow_5sec kopiér1OD3.png}} 
%\caption{(a): Total Flow difference; (b) Flow share difference, example OD-relation }
%\label{fig:LargeScale_total}
%\end{figure}

We scale the demands of the 8,046 OD-relations used in \cite{fosgerau_perturbed_2022} (stemming from the GPS observations) to the overall demand level in the case study area. Using these demands, the estimated route-choice parameters, and the BPR-based link cost-time functions defined for each network link, we calculate the network equilibrium flows. This is done by the approach introduced in \cite{yao_perturbed_2024}, obtaining the resulting link flow vector $x^*$ for each OD-relation. 
Using this, we compute the Jacobian of the equilibrium link flows as follows. First, the per-OD flow sensitivities with respect to equilibrium costs are computed using~\eqref{eq:eqngradient}. This involves evaluating 8046 matrices of size $30773\times{}30773$. Then, the sensitivities of equilibrium costs with respect to parameters are computed using~\eqref{eq:jac_eq_cost}. Finally, we compute the sensitivity of aggregate equilibrium link flows with respect to parameters using the chain rule \eqref{eq:type_specific_sensitivity}-\eqref{eq:total_sensitivity}.

To illustrate the precision of sensitivity analysis carried out using Theorem \ref{thm:jacobian_equilibrium_cost}, we consider a slight change in the free-flow travel time $t_0$ of a network link, and thus calculate the Jacobian with respect to this parameter. Computing the Jacobian takes 134 seconds, and with that, we can easily approximate the network-wide impacts of any small change in the free-flow travel time on any network link in virtually no additional time. In the following, we consider a scenario in which the free-flow travel time is increased by 10\% in the southbound direction of Langebro, one of the four bridges connecting Zealand (northwest) to the island of Amager (southeast). The location of the bridge can be seen in Figures \ref{fig:flowdiff} and \ref{fig:flowdiffZoom}. The resulting link flows are approximated using \eqref{taylor}, where $\nabla x^*(\theta)$ is computed using Eq.~\eqref{eq:type_specific_sensitivity}-\eqref{eq:total_sensitivity} and $\epsilon_{\theta}$ is $0.1\cdot{}t_0$ on the southbound link on Langebro and 0 elsewhere.

To evaluate the approximation, we also compute the exact equilibrium solution for the scenario. This computation takes 52.6 minutes, which is about 24 times longer than when using the Jacobian from Theorem \ref{thm:jacobian_equilibrium_cost}. Figure \ref{fig:scatter} compares the approximate changes in link flows computed using the Jacobian to the exact changes in link flows computed using the exact equilibrium solution.  We observe that predictions are very close to the 45$^\circ$ line.

Figures \ref{fig:flowdiff} and \ref{fig:flowdiffZoom} compare the absolute and relative changes in link flows using the approximate and exact calculations. As also suggested by Figure \ref{fig:scatter}, the differences are very small. Furthermore, the patterns look plausible: when the cost on the southbound direction of the Langebro bridge increases, flows on this and the main corridors leading to and from it decrease, while flows increase on the main alternatives, including other bridges connecting Zealand to the island of Amager. The largest change in flow is on the southbound direction of Langebro, where flow reduces by 193.65 and 192.05 vehicles in the approximate and exact computations, respectively.

\begin{figure}[H]
\centering

\includegraphics[width=0.49\textwidth]{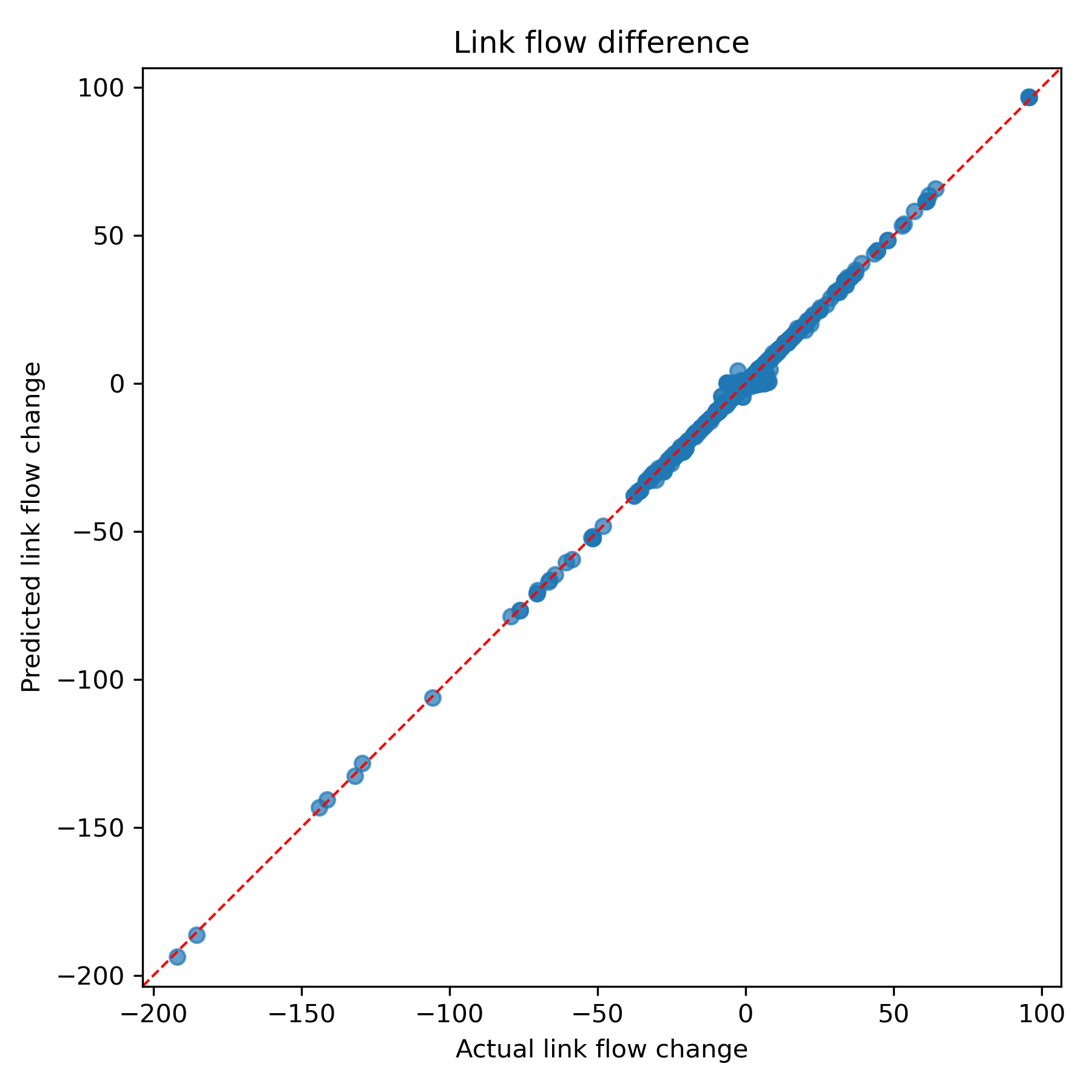} \hfill

\caption{Approximate link flow changes versus exact links flow changes [cars per day]. }
\label{fig:scatter}
\end{figure}

\begin{figure}[H]
\centering
\subfloat[\label{fig:flowdiffPred}]
{\includegraphics[width=0.49\textwidth]{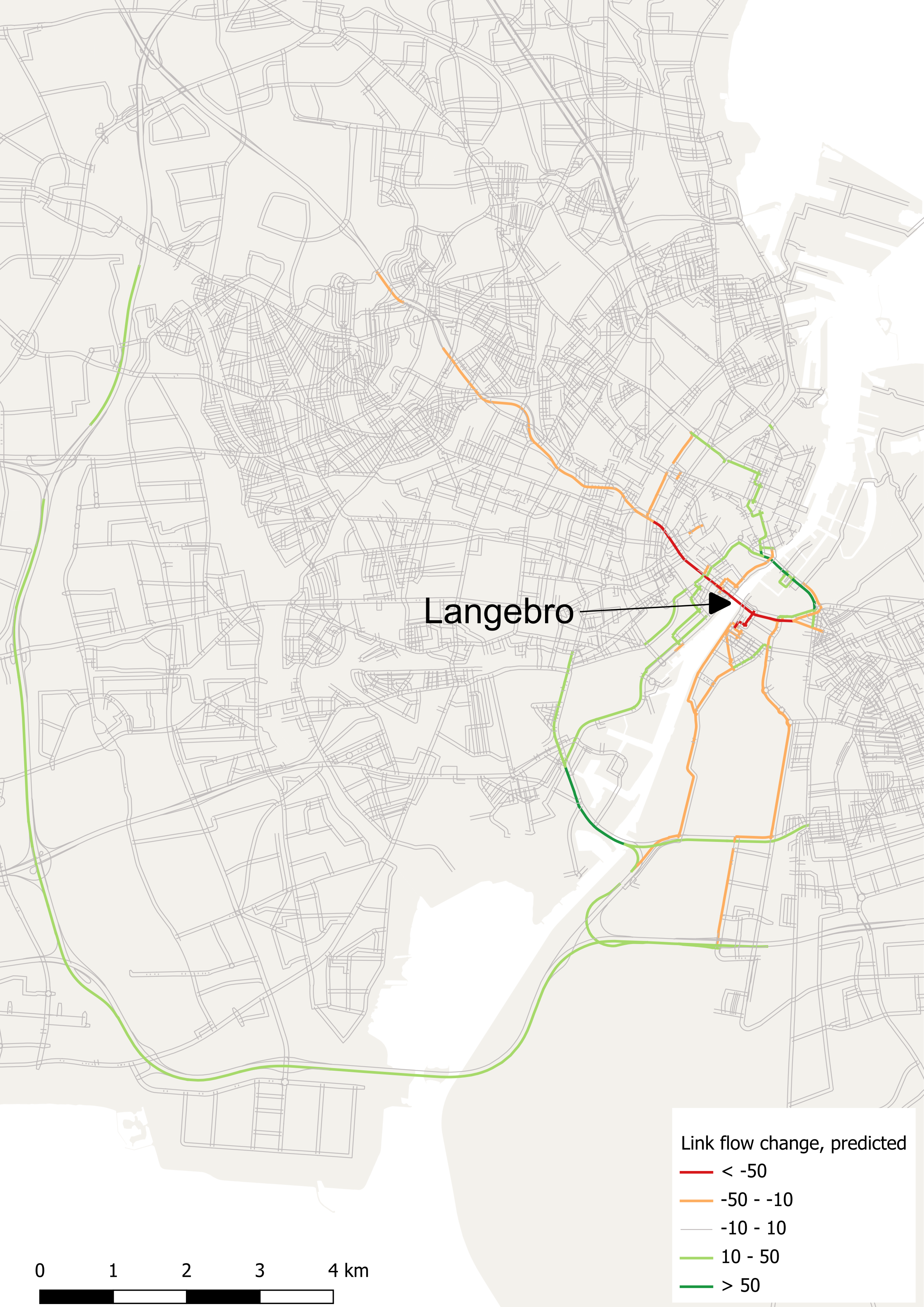}} \hfill
\subfloat[\label{fig:flowdiffActual}]{\includegraphics[width=0.49\textwidth]{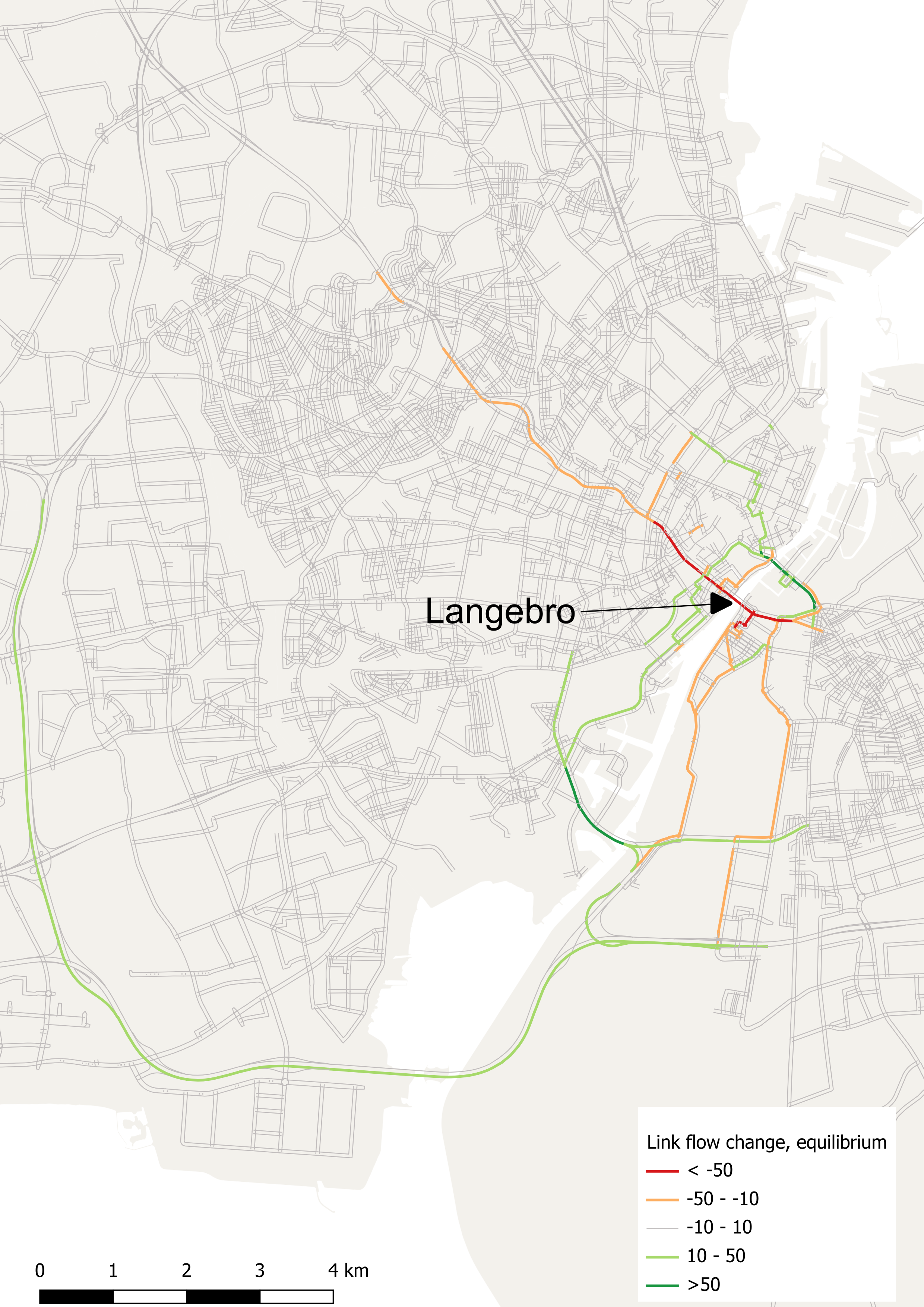}} 
\caption{(a) Approximate flow change, 10\% increase in free-flow travel time on southbound Langebro; (b) Exact flow change, 10\% increase in free-flow travel time on southbound Langebro.}
\label{fig:flowdiff}
\end{figure}

\begin{figure}[H]
\centering
\subfloat[\label{fig:flowdiffPred_2}]
{\includegraphics[width=0.49\textwidth]{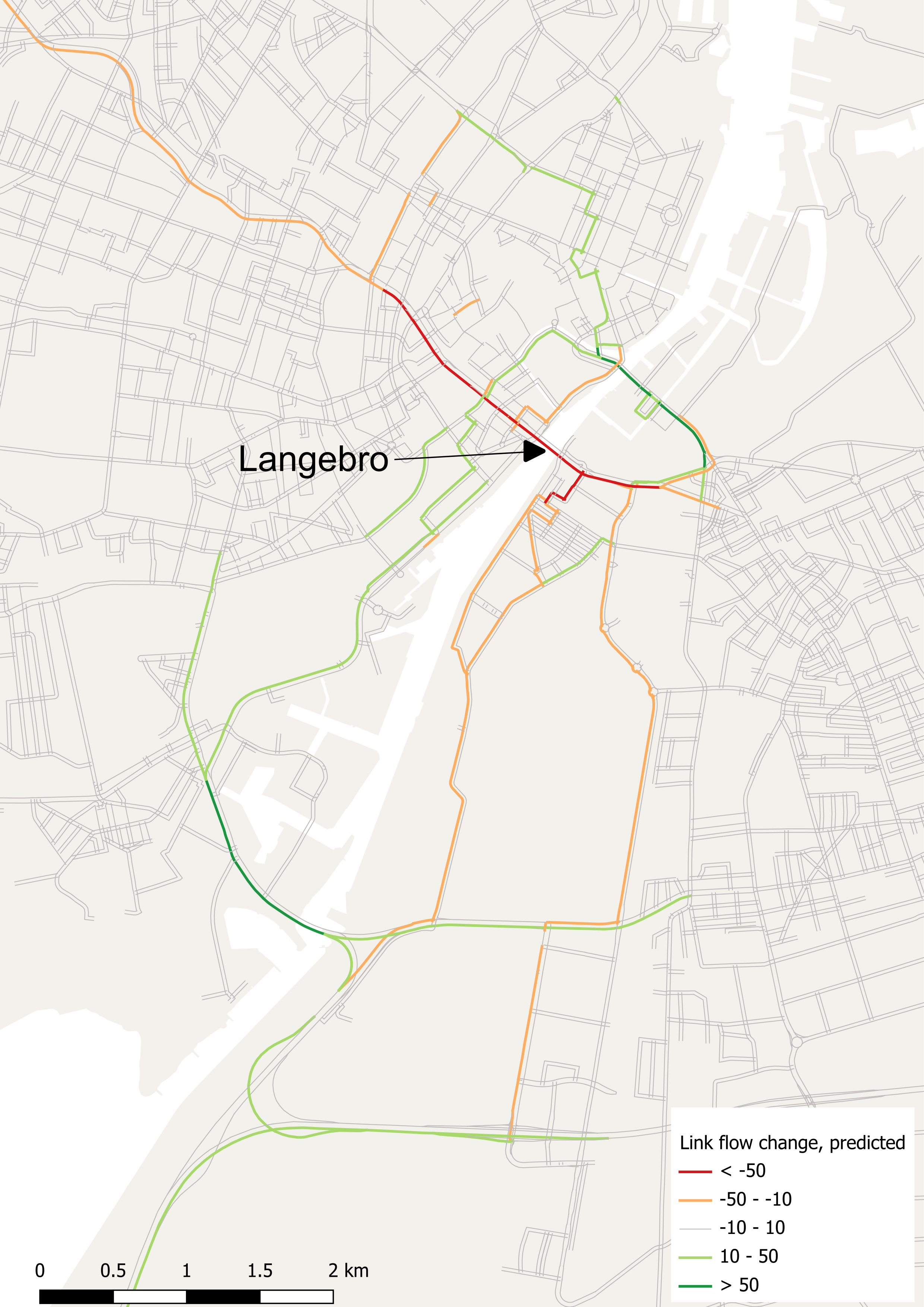}} \hfill
\subfloat[\label{fig:flowdiffActual_2}]{\includegraphics[width=0.49\textwidth]{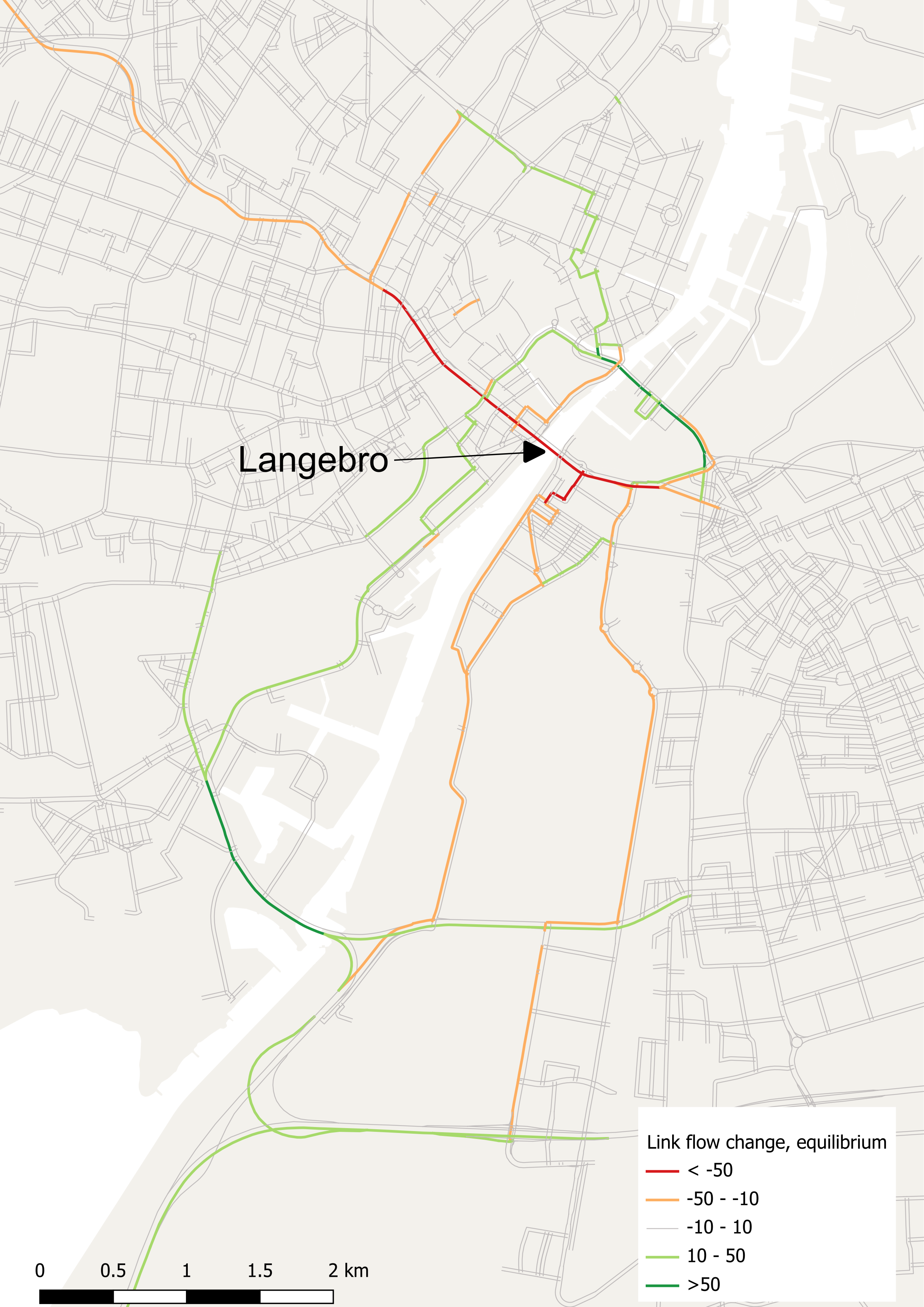}} 
\caption{(a) Approximate flow change, 10\% increase in free-flow travel time on southbound Langebro, zoom; (b) Exact flow change, 10\% increase in free-flow travel time on southbound Langebro, zoom.}
\label{fig:flowdiffZoom}
\end{figure}

In Table \ref{tab:WelfareMeasures}, we compute the welfare measure \eqref{eq:W} using both the approximation and the exact equilibrium solutions, as well as for the baseline. The welfare measure for the baseline is computed as,
\begin{equation}
    W = -\sum_w q^w \left[\left({x^{w*}}\right)^\top c^* + F(x^{w*})\right],
\end{equation}
and so is the scenario computation using the exact equilibrium solution. 
The approximate welfare measure for the scenario  is computed using equation \eqref{eq:welfare_approx} and adding the welfare of the baseline. The relative difference between the welfare measures of the approximate and the exact methods is very small ($2.6 \times 10^{-6}$).

\begin{table}[ht]
    \centering
    \begin{tabular}{rr} \hline
    & $W$\\ \hline
       Baseline  &  -14,658,895.8\\
       Scenario, approximate   & -14,659,681.7 \\
       Scenario, exact  &  -14,659,643.6\\ \hline
    \end{tabular} 
    \caption{Welfare measures for the baseline equilibrium solution and the scenario equilibrium solution calculated with the approximate and exact method, respectively}
    \label{tab:WelfareMeasures}
\end{table}

\section{Conclusion}\label{sec:conclusion}

This paper contributes to the literature on route choice and network equilibrium by developing a sensitivity analysis framework for the perturbed utility route choice (PURC) model. By deriving closed‐form expressions for the Jacobian of optimal PURC flows with respect to link cost parameters, our analysis is not only useful for improving computational efficiency but also provides deeper insights into the substitution and complementarity patterns inherent in PURC. In contrast to traditional additive random utility models that yield only substitution effects, our results reveal that complementarity among routes may emerge in the PURC model. 

Furthermore, the extension to flow-dependent cases allows us to estimate the changes in equilibrium link flows when network parameters are shifted. The numerical examples demonstrate how sensitivity measures can be used to predict changes in equilibrium flows and corresponding changes in welfare without resolving the entire traffic equilibrium problem, as well as to quantify the uncertainty in model prediction outputs resulting from uncertainties in, e.g., estimated link cost parameters. As our large-scale case study clearly demonstrates, this is also feasible for very large networks.

In sum, our framework offers a versatile tool for analyzing the response of travelers' routing decisions in the network to parameter shifts, thereby bridging a critical gap between theoretical modeling and empirical application. 
In ongoing research, we are using Theorem \ref{thm:jacobian} to develop a regression-based estimator for the PURC model that is applicable to individual-level data. This will widen the scope for application of the PURC model considerably, as it then becomes unnecessary to aggregate trips into observed origin-destination flows as was done in \citet{fosgerau_perturbed_2022}. Moreover, individual-level estimation allows individual-level information, such as income, to be incorporated. Future research could also apply the proposed approach to real-world network design and pricing problems.

\pagebreak 

\appendix
\section{Appendix}

The following lemma establishes mathematical properties of the projection matrix $P^*$ discussed in Section \ref{sec:PURC model}. 

\begin{lemma}\label{lem:P_hat}
Define 
\begin{equation*}
P^*=B^*-(AB^*)^+ A B^*,
\end{equation*}%
Then ${P^*}$ is an orthogonal projection matrix on the linear subspace $\{x\in \R^{|\mathcal{L}|}:Ax=0,B^* x=x\}$. Moreover, 
\begin{itemize}
    \item[a)] ${P^*}={P^*}{B^*}$;
    \item[b)] $ P^* =P^{* \top} =B^{*} -B^* A^\top \left( B^*     A^\top\right)^+$.
    \item[c)]${P^*}A^{\top }=0$;
    \item[d)] $ (P^*)^2 =P^*$ ;
\end{itemize}
\end{lemma}

\begin{proof}[Proof of Lemma \ref{lem:P_hat}]                \;

\begin{itemize}
    \item[a)] It follows by the definition of $P^*$ that $ P^* B^* =P^* $.
    \item[b)]  ${P^*}$ is symmetric by the definition of the Moore-Penrose inverse. Then point b) follows by straightforward calculation, using that Moore-Penrose inversion commutes with transposition.  
    \item[c)] Note that 
    \begin{eqnarray*}
        {P^*}A^{\top } &=& {P^*}^\top B^* A^{\top } \\
        &=&    \left(B^* -B^* A^\top \left( B^*      A^\top\right)^+\right)B^* A^{\top }   \\
        &=& B A^{\top } -B^* A^\top \left( B^*    A^\top\right)^+B^* A^{\top } =0.
    \end{eqnarray*}
    
    \item[d)] Follows once we have shown that $ P^*$ is an orthogonal projection.
\end{itemize}
 
We will show that $ P^* x$ minimizes the Euclidian distance to the linear subspace $S=\{x\in \R^{|\mathcal{L}|}:Ax=0,B^* x=x\}$. This is a quadratic problem with a unique solution $y$, the orthogonal projection of $x$ on $S$, given by the first-order condition
\[
0 = y - x + A^\top \eta + (I- B^*)\mu,
\]
where $\eta$ and $\mu$ are dual variables corresponding to the two subspace constraints. Applying $ P^*$ to the first-order condition and using the properties just established shows that $ P^* y =P^* x$. Next, $y\in S$ implies that $y =P^* x$. This shows that $P^*$ is the orthogonal projection on $S$.
\end{proof}

The following lemma is used in the proof of Theorem \ref{thm:jacobian}.

\begin{lemma}\label{lem:fullrank}
 Let $A\in \mathbb R^{n\times m}$ be an arbitrary matrix and $b\in \R^n$ . Suppose $Ax=b$ has at least one solution $x^*$. Then there exists $C,d$ such that \textit{i)} $Cx=d$ if and only if $Ax=b$ and \textit{ii)} $C$ has full rank.
\end{lemma}
\begin{proof}
    Let $r$ denote the rank of $A$ and let $A=U_rD_rV_r^\top$ denote the compact singular value decomposition of $A$, where $D_r$ is an $r\times r$ diagonal matrix of the nonzero singular values of $A$, $U_r$ is $n\times r$ and $V_r$ is $r\times m$ with $U_r^\top U_r =V_r^\top V_r =I_r$. 

   Now let $C=D_rV_r^\top$ and $d=U_r^\top b$. If $Ax=b$ then
    \begin{eqnarray}
        U_r D_r V_r^\top x =b \Rightarrow U_r^\top U_rD_r V_r^\top x=U_r^\top b \Rightarrow D_r V_r^\top x=U_r^\top b \Rightarrow Cx=d,
    \end{eqnarray}
using that $U_r^\top U_r=I_r$. On the other hand, if $Cx=d$ then
\begin{eqnarray}
 Cx=d\Rightarrow   U_r Cx=U_r d\Rightarrow Ax=U_rU_r^\top b.
\end{eqnarray}
  Note that the pseudo-inverse $A^+ =V_r D_r^{-1}U_r^\top$, such that $A^+A=V_rV_r^\top $ and $AA^+=U_r U_r^\top$. Therefore, $U_r U_r^\top b=AA^+b$. 
  Finally, since there exists $x^*$ such that $Ax^*=b$ we have that $AA^+b=AA^+Ax^*=Ax^*=b$, by properties of pseudo inverse $AA^+A = A$. Hence, we have that $Cx=d$ implies $Ax=b$. This completes the proof.
\end{proof}

\section{O-minimal structures and definability}\label{app:definable}

This section gives a very brief introduction to o-minimal structures and the concept of definability, leading up to a lemma showing that the optimal flow mapping $c\rightarrow x^*(c)$ is definable. We use this lemma to establish the claim in Theorem \ref{thm:jacobian} that the activation boundary $\mathcal C_0$ is a  Lebesgue null set. 

An o-minimal structure over the real line is a collection $\mathcal S=(\mathcal S^n)_{n\in \mathbb N}$, where each $\mathcal S^n$ is a set of subsets of $\R^n$ satisfying certain axioms \citep[][Def. 1.4]{coste1999introduction}. The collection $\mathcal S$ is constructed to be large enough that much interesting mathematical analysis is possible within this universe, but small enough that pathologies such as the Cantor sets are ruled out; it provides a "tame" topology \citep{van1998tame}. 

The elements of $\mathcal S^n$ are called the \emph{definable} subsets of $\R^n$. A map $f:A\rightarrow \R^p$ (where $A \subseteq \R^n$) is called definable if its graph is a definable subset of $\R^{n+p}$.

In general, the elements of $\mathcal S^1$ are precisely the finite unions of points and intervals. A variety of results exists for constructing definable sets and functions within an o-minimal structure. 

Several o-minimal structures are known. The most relevant for our purposes is $\mathbb{R}_{\mathrm{exp}}$, which allows perturbation functions constructed from polynomials and the exponential function~\citep[][p.67]{coste1999introduction}. The graph of the logarithm function is definable as it is the same as the graph of the exponential, with coordinates swapped. Hence, the logarithm function is also allowed. 

The proof of Theorem \ref{thm:jacobian}, item 3, refers to the following lemma.

\begin{lemma}\label{lem:definable}
    Fix an o-minimal structure on $\R$. If the perturbation function $F$ is definable in this structure, then the optimal flow  $x^*$ is a  definable function of $c$.
\end{lemma}
\begin{proof}
Since $F$ is definable, the PURC objective $c^\top x +F(x)$ is also definable. The constraint that the flow belongs to the set $D:=\{x: Ax=b, x \geq 0\}$ is also definable.

Consider the graph of the $\argmin$ correspondence:
\[
\Gamma_{\argmin} = \left\{(x,c)\in \mathbb{R}^{|\mathcal{L}|}\times \mathbb{R}^{|\mathcal{L}|} : x\in D, \forall z \in D: c^\top x +F(x) \leq c^\top z +F(z) \right\}. 
\]
This is definable, since $D$ and the PURC objective are definable. Hence, the  $\argmin$ correspondence is definable.

Now, the optimal flow function $x^*(c)$ is the unique minimizer of the PURC objective and hence it is definable. 
\end{proof}

\bibliographystyle{jfe}
\bibliography{MFzotero,mybib, bibTkra}

@article{fosgerau_perturbed_2022,
	title = {A perturbed utility route choice model},
	volume = {136},
	doi = {https://doi.org/10.1016/j.trc.2021.103514},
	number = {103514},
	journal = {Transportation Research Part C},
	author = {Fosgerau, Mogens and Paulsen, Mads and Rasmussen, Thomas Kjær},
	year = {2022},
}

@incollection{fosgerau_rule---half_2021,
	address = {Oxford},
	title = {The {Rule}-of-a-{Half} and {Interpreting} the {Consumer} {Surplus} as {Accessibility}},
	isbn = {978-0-08-102672-4},
	url = {https://www.sciencedirect.com/science/article/pii/B9780081026717100454},
	urldate = {2025-09-04},
	booktitle = {International {Encyclopedia} of {Transportation}},
	publisher = {Elsevier},
	author = {Fosgerau, Mogens and Pilegaard, Ninette},
	editor = {Vickerman, Roger},
	month = jan,
	year = {2021},
	doi = {10.1016/B978-0-08-102671-7.10045-4},
	pages = {237--241},
}

@article{Fosgerau2013y,
	title = {Choice probability generating functions},
	volume = {8},
	issn = {17555345},
	doi = {10.1016/j.jocm.2013.05.002},
	journal = {Journal of Choice Modelling},
	author = {Fosgerau, Mogens and McFadden, D. and Bierlaire, M.},
	year = {2013},
}

@article{yao_perturbed_2024,
	title = {Perturbed utility stochastic traffic assignment},
	volume = {58},
	number = {4},
	journal = {Transportation Science},
	author = {Yao, Rui and Fosgerau, Mogens and Paulsen, M. and Rasmussen, T.K.},
	year = {2024},
}

@article{fosgerau_bikeability_2023,
	title = {Bikeability and the induced demand for cycling},
	doi = {https://doi.org/10.1073/pnas.2220515120},
	number = {120(16)},
	journal = {Proceedings of the National Academy of Sciences},
	author = {Fosgerau, Mogens and Łukawska, M. and Paulsen, M. and Rasmussen, T.K.},
	year = {2023},
}

@book{rockafellar_convex_1970,
	address = {Princeton, N.J.},
	title = {Convex {Analysis}},
	isbn = {0-691-08069-0},
	publisher = {Princeton University Press},
	author = {Rockafellar, R. Tyrrell},
	year = {1970},
}

@book{krantz_implicit_2003,
	address = {Boston, MA},
	title = {The {Implicit} {Function} {Theorem}},
	copyright = {http://www.springer.com/tdm},
	isbn = {978-1-4612-6593-1 978-1-4612-0059-8},
	url = {http://link.springer.com/10.1007/978-1-4612-0059-8},
	language = {en},
	urldate = {2025-04-14},
	publisher = {Birkhäuser},
	author = {Krantz, Steven G. and Parks, Harold R.},
	year = {2003},
	doi = {10.1007/978-1-4612-0059-8},
}

@article{Fosgerau2021,
	title = {Some {Remarks} on {CCP}-based {Estimators} of {Dynamic} {Models}},
	volume = {204},
	doi = {doi.org/10.1016/j.econlet.2021.109911},
	number = {109911},
	journal = {Economics Letters},
	author = {Fosgerau, Mogens and Melo, Emerson and Shum, Matthew and Sørensen, Jesper R.-V.},
	year = {2021},
}

@article{patriksson_sensitivity_2004,
	title = {Sensitivity {Analysis} of {Traffic} {Equilibria}},
	volume = {38},
	issn = {0041-1655},
	url = {https://pubsonline.informs.org/doi/10.1287/trsc.1030.0043},
	doi = {10.1287/trsc.1030.0043},
	number = {3},
	urldate = {2024-12-17},
	journal = {Transportation Science},
	author = {Patriksson, Michael},
	month = aug,
	year = {2004},
	note = {Publisher: INFORMS},
	pages = {258--281},
}

@article{Sorensen2019,
	title = {How {McFadden} met {Rockafellar} and learned to do more with less},
	volume = {100},
	doi = {https://doi.org/10.1016/j.jmateco.2021.102629},
	journal = {Journal of Mathematical Economics},
	author = {Sørensen, Jesper R.-V. and Fosgerau, Mogens},
	year = {2022},
}

@incollection{mcfadden_econometric_1981,
	address = {Cambridge, MA, USA},
	title = {Econometric {Models} of {Probabilistic} {Choice}},
	isbn = {0-262-13159-5},
	booktitle = {Structural {Analysis} of {Discrete} {Data} with {Econometric} {Applications}},
	publisher = {MIT Press},
	author = {McFadden, Daniel},
	editor = {Manski, C. and McFadden, D.},
	year = {1981},
	doi = {10.1086/296093},
	note = {ISSN: 00219398},
	pages = {198--272},
}

@article{milgrom_envelope_2002,
	title = {Envelope {Theorems} for {Arbitrary} {Choice} {Sets}},
	volume = {70},
	issn = {0012-9682},
	url = {http://dx.doi.org/10.1111/1468-0262.00296},
	doi = {10.1111/1468-0262.00296},
	number = {2},
	journal = {Econometrica},
	author = {Milgrom, Paul and Segal, Ilya},
	year = {2002},
	note = {ISBN: 00129682},
	pages = {583--601},
}

@book{Small1992,
	address = {London and New York},
	title = {Urban transportation economics},
	publisher = {Harwood Academic Publishers},
	author = {Small, Kenneth A. and Verhoef, Erik T.},
	year = {2007},
}

@article{Alexandroff1939,
	title = {Almost everywhere existence of the second differential of a convex function and some properties of convex surfaces connected with it},
	volume = {6},
	journal = {Uchenye Zapiski Leningrad Gos. Univ., Math. Ser},
	author = {Aleksandrov, A D},
	year = {1939},
	pages = {3--35},
}

@article{yang2013sensitivity,
  title={Sensitivity-based uncertainty analysis of a combined travel demand model},
  author={Yang, Chao and Chen, Anthony and Xu, Xiangdong and Wong, SC},
  journal={Transportation Research Part B: Methodological},
  volume={57},
  pages={225--244},
  year={2013},
  publisher={Elsevier}
}

@article{josefsson2007sensitivity,
  title={Sensitivity analysis of separable traffic equilibrium equilibria with application to bilevel optimization in network design},
  author={Josefsson, Magnus and Patriksson, Michael},
  journal={Transportation Research Part B: Methodological},
  volume={41},
  number={1},
  pages={4--31},
  year={2007},
  publisher={Elsevier}
}

@article{liu2022inducing,
  title={Inducing equilibria via incentives: Simultaneous design-and-play ensures global convergence},
  author={Liu, Boyi and Li, Jiayang and Yang, Zhuoran and Wai, Hoi-To and Hong, Mingyi and Nie, Yu and Wang, Zhaoran},
  journal={Advances in Neural Information Processing Systems},
  volume={35},
  pages={29001--29013},
  year={2022}
}

@article{yang1997traffic,
  title={Traffic restraint, road pricing and network equilibrium},
  author={Yang, Hai and Bell, Michael GH},
  journal={Transportation Research Part B: Methodological},
  volume={31},
  number={4},
  pages={303--314},
  year={1997},
  publisher={Elsevier}
}

@inproceedings{franceschi2018bilevel,
  title={Bilevel programming for hyperparameter optimization and meta-learning},
  author={Franceschi, Luca and Frasconi, Paolo and Salzo, Saverio and Grazzi, Riccardo and Pontil, Massimiliano},
  booktitle={International conference on machine learning},
  pages={1568--1577},
  year={2018},
  organization={PMLR}
}

@inproceedings{lorraine2020optimizing,
  title={Optimizing millions of hyperparameters by implicit differentiation},
  author={Lorraine, Jonathan and Vicol, Paul and Duvenaud, David},
  booktitle={International conference on artificial intelligence and statistics},
  pages={1540--1552},
  year={2020},
  organization={PMLR}
}

@article{bai2019deep,
  title={Deep equilibrium models},
  author={Bai, Shaojie and Kolter, J Zico and Koltun, Vladlen},
  journal={Advances in neural information processing systems},
  volume={32},
  year={2019}
}

@article{robinson1985implicit,
  title={Implicit B-differentiability in generalized equations(Technical Summary Report)},
  author={Robinson, Stephen M},
  year={1985}
}

@article{robinson2006strong,
  title={Strong regularity and the sensitivity analysis of traffic equilibria: A comment},
  author={Robinson, Stephen M},
  journal={Transportation Science},
  volume={40},
  number={4},
  pages={540--542},
  year={2006},
  publisher={INFORMS}
}

@article{du2022sensitivity,
  title={Sensitivity analysis for transit equilibrium assignment and applications to uncertainty analysis},
  author={Du, Muqing and Chen, Anthony},
  journal={Transportation Research Part B: Methodological},
  volume={157},
  pages={175--202},
  year={2022},
  publisher={Elsevier}
}

@article{patriksson2003sensitivity,
  title={Sensitivity analysis of aggregated variational inequality problems, with application to traffic equilibria},
  author={Patriksson, Michael and Rockafellar, R Tyrrell},
  journal={Transportation Science},
  volume={37},
  number={1},
  pages={56--68},
  year={2003},
  publisher={INFORMS}
}

@book{clarke1990optimization,
  title={Optimization and nonsmooth analysis},
  author={Clarke, Frank H},
  year={1990},
  publisher={SIAM}
}

@article{cho2000reduction,
  title={A reduction method for local sensitivity analyses of network equilibrium arc flows},
  author={Cho, Hsun-Jung and Smith, Tony E and Friesz, Terry L},
  journal={Transportation Research Part B: Methodological},
  volume={34},
  number={1},
  pages={31--51},
  year={2000},
  publisher={Elsevier}
}

@article{baillon2008markovian,
  title={Markovian traffic equilibrium},
  author={Baillon, J-B and Cominetti, Roberto},
  journal={Mathematical Programming},
  volume={111},
  number={1},
  pages={33--56},
  year={2008},
  publisher={Springer}
}

@article{maher1997probit,
  title={A probit-based stochastic user equilibrium assignment model},
  author={Maher, MJ and Hughes, PC},
  journal={Transportation Research Part B: Methodological},
  volume={31},
  number={4},
  pages={341--355},
  year={1997},
  publisher={Elsevier}
}

@article{yang2009sensitivity,
  title={Sensitivity analysis of the combined travel demand model with applications},
  author={Yang, Chao and Chen, Anthony},
  journal={European Journal of Operational Research},
  volume={198},
  number={3},
  pages={909--921},
  year={2009},
  publisher={Elsevier}
}

@article{ying2005sensitivity,
  title={Sensitivity analysis of stochastic user equilibrium flows in a bi-modal network with application to optimal pricing},
  author={Ying, Jiang Qian and Yang, Hai},
  journal={Transportation research Part B: methodological},
  volume={39},
  number={9},
  pages={769--795},
  year={2005},
  publisher={Elsevier}
}

@article{clark2006applications,
  title={Applications of sensitivity analysis for probit stochastic network equilibrium},
  author={Clark, Stephen D and Watling, David P},
  journal={European journal of operational research},
  volume={175},
  number={2},
  pages={894--911},
  year={2006},
  publisher={Elsevier}
}

@book{biggs1993algebraic,
  title={Algebraic graph theory},
  author={Biggs, Norman},
  number={67},
  year={1993},
  publisher={Cambridge university press}
}

@book{golub2013matrix,
  title={Matrix computations},
  author={Golub, Gene H and Van Loan, Charles F},
  year={2013},
  publisher={JHU press}
}

@article{lu2008sensitivity,
  title={Sensitivity of static traffic user equilibria with perturbations in arc cost function and travel demand},
  author={Lu, Shu},
  journal={Transportation science},
  volume={42},
  number={1},
  pages={105--123},
  year={2008},
  publisher={INFORMS}
}

@book{saad2003iterative,
  title={Iterative methods for sparse linear systems},
  author={Saad, Yousef},
  year={2003},
  publisher={SIAM}
}

@book{luo1996mathematical,
  title={Mathematical programs with equilibrium constraints},
  author={Luo, Zhi-Quan and Pang, Jong-Shi and Ralph, Daniel},
  year={1996},
  publisher={Cambridge University Press}
}

@article{gao2004continuous,
  title={A continuous equilibrium network design model and algorithm for transit systems},
  author={Gao, Ziyou and Sun, Huijun and Shan, Lian Long},
  journal={Transportation Research Part B: Methodological},
  volume={38},
  number={3},
  pages={235--250},
  year={2004},
  publisher={Elsevier}
}

@article{yang1998models,
  title={Models and algorithms for road network design: a review and some new developments},
  author={Yang, Hai and H. Bell, Michael G},
  journal={Transport Reviews},
  volume={18},
  number={3},
  pages={257--278},
  year={1998},
  publisher={Taylor \& Francis}
}

@article{yan1996optimal,
  title={Optimal road tolls under conditions of queueing and congestion},
  author={Yan, Hai and Lam, William HK},
  journal={Transportation Research Part A: Policy and Practice},
  volume={30},
  number={5},
  pages={319--332},
  year={1996},
  publisher={Elsevier}
}

@article{wang2021optimal,
  title={Optimal toll design problems under mixed traffic flow of human-driven vehicles and connected and autonomous vehicles},
  author={Wang, Jian and Lu, Lili and Peeta, Srinivas and He, Zhengbing},
  journal={Transportation Research Part C: Emerging Technologies},
  volume={125},
  pages={102952},
  year={2021},
  publisher={Elsevier}
}

@article{boyles2020transportation,
  title={Transportation network analysis},
  author={Boyles, Stephen D and Lownes, Nicholas E and Unnikrishnan, Avinash},
  journal={Volume I: Static and Dynamic Traffic Assignment},
  year={2020}
}

@book{coste1999introduction,
  title={An introduction to o-minimal geometry},
  author={Coste, Michel},
  year={1999}
}

@book{van1998tame,
  title={Tame topology and o-minimal structures},
  author={Van den Dries, Lou},
  volume={248},
  year={1998},
  publisher={Cambridge university press}
}

@book{federer2014geometric,
  title={Geometric measure theory},
  author={Federer, Herbert},
  year={2014},
  publisher={Springer}
}

@article{oyama2022markovian,
  title={Markovian traffic equilibrium assignment based on network generalized extreme value model},
  author={Oyama, Yuki and Hara, Yusuke and Akamatsu, Takashi},
  journal={Transportation Research Part B: Methodological},
  volume={155},
  pages={135--159},
  year={2022},
  publisher={Elsevier}
}

@techreport{beckmann1956studies,
  title={Studies in the Economics of Transportation},
  author={Beckmann, Martin and McGuire, Charles B and Winsten, Christopher B},
  year={1956}
}

@article{bekhor2001stochastic,
  title={Stochastic user equilibrium formulation for generalized nested logit model},
  author={Bekhor, S and Prashker, JN},
  journal={Transportation Research Record},
  volume={1752},
  number={1},
  pages={84--90},
  year={2001},
  publisher={SAGE Publications Sage CA: Los Angeles, CA}
}

@article{kitthamkesorn2014unconstrained,
  title={Unconstrained weibit stochastic user equilibrium model with extensions},
  author={Kitthamkesorn, Songyot and Chen, Anthony},
  journal={Transportation Research Part B: Methodological},
  volume={59},
  pages={1--21},
  year={2014},
  publisher={Elsevier}
}

@article{dial1971probabilistic,
  title={A probabilistic multipath traffic assignment model which obviates path enumeration},
  author={Dial, Robert B},
  journal={Transportation research},
  volume={5},
  number={2},
  pages={83--111},
  year={1971},
  publisher={Pergamon}
}

@article{fosgerau2012theory,
  title={A theory of the perturbed consumer with general budgets},
  author={Fosgerau, Mogens and McFadden, Daniel L},
  journal={NBER Working Paper},
  volume={17953},
  year={2012}
}

@article{allen2019identification,
  title={Identification with additively separable heterogeneity},
  author={Allen, Roy and Rehbeck, John},
  journal={Econometrica},
  volume={87},
  number={3},
  pages={1021--1054},
  year={2019},
  publisher={Wiley Online Library}
}

@article{tobin1988sensitivity,
  title={Sensitivity analysis for equilibrium network flow},
  author={Tobin, Roger L and Friesz, Terry L},
  journal={Transportation Science},
  volume={22},
  number={4},
  pages={242--250},
  year={1988},
  publisher={INFORMS}
}

\end{document}